%% file: main.tex
\documentclass[acmsmall]{acmart}

\usepackage{algorithmicx}
\usepackage{algpseudocode}
\usepackage{algorithm}
\algdef{SE}[DOWHILE]{Do}{doWhile}{\algorithmicdo}[1]{\algorithmicwhile\ #1}%
\algrenewcommand\algorithmicrequire{\textbf{Input:}}
\algrenewcommand\algorithmicensure{\textbf{Output:}}
\algnewcommand\algorithmicinitialization{\textbf{Initialization:}}
\algnewcommand\Initialization{\item[\algorithmicinitialization]}%

\usepackage{amsthm,amsmath,amssymb}
\usepackage{hyperref}
\usepackage{xspace}

\usepackage{xcolor}
\usepackage{soul}

\def \H {\ensuremath{\mathcal{H}}\xspace}
\def \I  {\ensuremath{\mathcal{I}}\xspace}
\def \F {\ensuremath{\mathcal{F}}\xspace}
\def \V {\ensuremath{\mathcal{V}}\xspace}

\def \C {\ensuremath{\mathcal{C}}\xspace}

\def \X {\ensuremath{\mathcal{X}}\xspace}

\def \Asgn {\mathit{asgn}}

\def \G {\mathit{g}}

\def \Trans {\mathit{Trans}}
\def \Inv {\mathit{Inv}}
\def \Init {\mathit{Init}}
\def \Flow {\mathit{Flow}}

\theoremstyle{definition}
\newtheorem{definition}{Definition}
\newtheorem{proposition}{Proposition}
\newtheorem{claim}{Claim}

\usepackage{multirow}
\usepackage{multicol}
\usepackage{makecell}
\usepackage{graphicx}
\usepackage{caption}
\usepackage{subcaption}
\usepackage{wrapfig,lipsum,booktabs}
\usepackage{adjustbox}

\usepackage{wrapfig}
\usepackage{colortbl} 

\title{A Cegar-centric Bounded Reachability Analysis for Compositional Affine Hybrid Systems}

\author{Atanu Kundu}
\affiliation{\institution{Indian Association for the Cultivation of Science} \country{India}}
\email{mcsak2346@iacs.res.in}

\author{Pratyay Sarkar}
\affiliation{\institution{Indian Association for the Cultivation of Science}\country{India} }
\email{ugps2578@iacs.res.in}

\author{Rajarshi Ray}
\affiliation{\institution{Indian Association for the Cultivation of Science}\country{India}}
\email{rajarshi.ray@iacs.res.in}

\begin{document}

\input{abstract}
\maketitle

\input{introduction}
\input{preliminaries}

\input{methodology}
\input{results}
\input{conclusion}

\bibliographystyle{plain}

\input{output}
\end{document}

%% file: abstract.tex
\begin{abstract}
Reachability analysis of compositional hybrid systems, where individual components are modeled as hybrid automata, poses unique challenges. In addition to preserving the compositional semantics while computing system behaviors, algorithms have to cater to the explosion in the number of locations in the parallel product automaton. In this paper, we propose a bounded reachability analysis algorithm for compositional hybrid systems with piecewise affine dynamics, based on the principle of counterexample guided abstraction refinement (CEGAR). In particular, the algorithm searches for a counterexample in the discrete abstraction of the composition model, without explicitly computing a product automaton. When a counterexample is discovered in the abstraction, its validity is verified by a refinement of the state-space guided by the abstract counterexample. The state-space refinement is through a symbolic reachability analysis, particularly using a state-of-the-art algorithm with support functions as the continuous state representation. In addition, the algorithm mixes different semantics of composition with the objective of improved efficiency. Step compositional semantics is followed while exploring the abstract (discrete) state-space, while shallow compositional semantics is followed during state-space refinement with symbolic reachability analysis. Optimizations such as caching the results of the symbolic reachability analysis, which can be later reused, have been proposed. We implement this algorithm in the tool \textsc{SAT-Reach} and demonstrate the scalability benefits.
 \end{abstract}

%% file: introduction.tex
\section{Introduction}
A hybrid automaton is a widely accepted framework for modeling and specification of systems that exhibit a combination of discrete and continuous dynamics, called hybrid systems \cite{alur1995algorithmic,alur1993,Henzinger-HA}. Complex hybrid systems are often designed as compositions of several simpler hybrid systems that interact with each other. Such systems are modeled by the parallel composition of hybrid automata representations of individual components \cite{alur1993}. Reachability analysis of compositional hybrid systems is challenging, not only due to the compositional semantics but also because of the explosion in the number of locations in the product automaton. Preserving synchronization semantics and consistency of shared variables across components poses unique challenges when computing model behaviors. Among the successful methods for formal verification of systems having a large discrete state space is bounded model checking (BMC) \cite{clarke2001bounded,armin-sat}. It is a technique in which a logic formula is constructed for the unwinding of the model up to $k$ next state transitions, conjoined with the negation of the given safety property $\phi$. The satisfiability of the constructed formula implies the existence of a counterexample trajectory in the model within $k$ transitions, whereas the unsatisfiability of the formula concludes the safety of the model concerning the property $\phi$ and the bound $k$. Several works have extended BMC of discrete state transition systems to timed and hybrid systems (\textsc{MathSat}~ \cite{audemard2002bounded}, \textsc{HySat}~ \cite{FMSD05}, \textsc{Bach} \cite{sat-lp-iis}, \textsc{dReach} \cite{kong2015dreach}, \textsc{iSatODE} \cite{eggers2011iSATODE}), reducing BMC to the satisfiability problem of Boolean combinations of arithmetic constraints on real-valued variables. Recent advances in the scalability of Boolean satisfiability (SAT) and SMT (Satisfiability Modulo Theory) solvers have significantly contributed to the success of BMC. In a different line of research, bounded reachability analysis algorithms for hybrid automata employ symbolic state representations to compute an over-approximation of the set of states reachable within a bounded number of discrete transitions in the automaton and within a bounded local time horizon \cite{DBLP:conf/hvc/RayGDBBG15}, \cite{chen2013flow}, \cite{10.1145/3567425}. 

\noindent \textbf{Our Contribution} We present a bounded reachability analysis algorithm for compositional hybrid systems with piecewise affine dynamics, based on the principle of CEGAR. The analysis is bounded on the number of discrete transitions in the automaton and the local time horizon (the time horizon to search for successor states within a location of the automaton). The algorithm initially identifies an abstract counterexample and subsequently verifies its validity through a symbolic state-space exploration procedure specifically designed for compositional affine hybrid dynamics. In addition, motivated by \cite{10.1007/978-3-030-94583-1_23}, the algorithm mixes different compositional semantics in various phases of the analysis to optimize efficiency while preserving soundness. Specifically, step semantics is employed in the abstract counterexample enumeration phase, while shallow compositional semantics is utilized in the refinement phase. The abstract counterexample enumeration phase finds compositional paths from satisfying solutions of a propositional logic formula encoding the discrete dynamics of the components and the synchronization semantics. The refinement phase verifies the abstract counterexample by incorporating a symbolic reachability analysis algorithm. Many step compositional paths, which are \emph{stutter variants} of each other, map to a single path in shallow semantics. Therefore, reachability analysis of a path in shallow semantics suffices for a reachability analysis of all mapping paths in step compositional semantics. 

In contrast to bounded reachability analysis tools for affine hybrid systems, which require product construction as a prerequisite when processing compositional systems, our algorithm does not explicitly compute the product automaton. In addition to avoiding explicit product construction, a path-oriented state-space exploration has been shown to guide the search toward state-space regions of interest with respect to a safety property \cite{10.1145/3567425}. A mixed semantics-guided exploration further enhances the performance. We implement this algorithm in the tool \textsc{SAT-Reach} and demonstrate the scalability benefits when compared to state-space exploration tools that require explicit product construction, the state-of-the-art reachability analysis tool for affine hybrid systems \textsc{SpaceEx} that implements a technique of on-the-fly product construction, \textsc{BACH}, a layered BMC tool for compositional linear systems incorporating SAT solving in discrete layer and linear constraint solving in the continuous layer, and \textsc{dReach} that is a BMC tool for general non-linear hybrid systems based on $\delta$-approximate satisfiability solving.

\noindent \textbf{Related Works} 
Bounded Model Checking (BMC) for hybrid systems has been extensively investigated, with numerous studies \cite{audemard2002bounded,FMSD05,eggers2011iSATODE,damm2012exact,sat-lp-iis,kong2015dreach,10.1007/978-3-662-46681-0_4,FLGDCRLRGDM11,10.1145/3567425,DBLP:conf/hvc/RayGDBBG15} proposing diverse strategies. Below, we discuss the existing algorithms for linear, affine, and nonlinear hybrid systems, highlighting methodological distinctions and contrasting them with our proposed approach.

\textbf{BMC for Linear Hybrid Systems}
\textsc{HySAT} \cite{FMSD05} is a bounded model checker for linear hybrid automata that integrates a SAT solver with a linear programming routine as its core engine. The bounded verification problem is encoded as a Boolean combination of linear arithmetic constraints on real and Boolean variables that represent the continuous and discrete components of the system. Furthermore, optimizations such as isomorphy inference, constraint sharing, and heuristics for variable selection order boost performance. A symbolic verification algorithm for linear hybrid automata with a large discrete state space is proposed in \cite{damm2012exact}. This algorithm emphasizes symbolic representation and efficient state-space traversal methods, eliminating redundancy and minimizing constraints. An SMT solver is utilized to detect and eliminate redundant linear constraints. \textsc{BACH}~\cite{10.1007/978-3-030-94583-1_23} proposes a layered BMC routine for compositional linear hybrid automata (CLHA). Instead of encoding the entire bounded state space into one large SMT formula, the proposed method divides the problem into discrete and continuous layers. Candidate unsafe paths are enumerated by sat-solving in the discrete layer, and then the feasibility of the candidate paths is checked with linear constraint solving. A mix of compositional semantics has been proposed for efficiently enumerating a potentially large set of paths. Similar to \textsc{BACH}, our proposed algorithm follows a CEGAR-based state-space exploration strategy, where abstract counterexamples are searched from the discrete abstraction, followed by a set-based symbolic reachability analysis algorithm as a state-space refinement. Unlike \textsc{Bach}, where feasibility checking of a path is reduced to solving a set of linear constraints, we check the feasibility of abstract counterexamples with an over-approximative flowpipe computation. This is mainly because encoding affine hybrid dynamics as constraints leads to non-linear real arithmetic constraints with transcendental functions. Satisfiability solving of such formulae is known to be undecidable.

\textbf{BMC for Affine Hybrid Systems}
\textsc{SpaceEx} \cite{FLGDCRLRGDM11} computes the set of reachable states, also called a flowpipe, for linear and affine hybrid systems. Compositional systems are analyzed by an on-the-fly exploration strategy in which only the reachable parts of the product automaton are created in memory \cite {Intro2SpaceEx}. \textsc{XSpeed} \cite{DBLP:conf/hvc/RayGDBBG15} can compute a flowpipe for affine hybrid systems using parallel implementations to extract performance from multicore processors. \textsc{Flow*} \cite{chen2013flow} computes flowpipe of nonlinear hybrid systems using Taylor Models for accurate symbolic reachability analysis, focusing on bounded-depth exploration. \textsc{SAT-Reach} \cite{10.1145/3567425} is a bounded reachability analysis tool based on the CEGAR framework for affine hybrid systems. \textsc{SAT-Reach} implements a synergistic combination of SAT-based path enumeration, path-wise reachability analysis by flowpipe construction, and searching for an unsafe execution by the trajectory splicing algorithm. It can generate a counterexample as proof of a safety violation. All the above tools can verify reachability properties. \textsc{HyComp} \cite{10.1007/978-3-662-46681-0_4} is an SMT-based model checker for hybrid systems built over \textsc{NuXmv} and \textsc{NuSMV} model checkers and integrates \textsc{MathSAT} as an SMT solver. It supports LHA, affine, and polynomial flow dynamics. The tool can analyze a network of hybrid automata. It relies on an encoding of the network that discretizes the dynamics into an infinite-state transition system, which is subsequently analyzed using SMT-based verification techniques in \textsc{NuXMV}, such as BMC, K-induction, and IC3 with predicate abstraction. The tool can verify invariant properties, LTL properties, and scenario specifications. Tools like \textsc{Flow*}, \textsc{XSpeed} and \textsc{SAT-Reach} require a product construction for the analysis of compositional hybrid systems.
In contrast, the proposed algorithm does not require an explicit product construction. The closest work related to ours is \textsc{SAT-Reach}. The major enhancements being (a) not requiring a product construction and  (b) a mixed-compositional semantics guided exploration for efficiency.

\textbf{BMC for Non-linear Hybrid Systems}
\textsc{dReach} \cite{kong2015dreach} is a bounded model checking tool for nonlinear hybrid systems. It encodes a bounded depth $\delta$-over-approximation of the hybrid system together with the specification as SMT constraints over Type-2 computable functions, and an SMT solver \textsc{dReal} is used to solve the constraints. \textsc{iSAT-ODE}~\cite{eggers2011iSATODE} is a bounded verification tool for nonlinear hybrid systems. It encodes the $k$-unwinding of system dynamics and the verification condition as a formula with mixed Boolean and arithmetic constraints, including ODEs that may involve transcendental functions. The key contribution is a satisfiability checking algorithm that integrates DPLL-style search with alternating assignment and deduction in the region space. The implementation of exact arithmetic with outward rounding of the interval rounding ensures a numerical over-approximation in the enclosure computation. A reachability analysis of non-linear analog circuits by linearization into piecewise linear hybrid automata with many locations is reported in \cite{7059096}. Scalability is achieved by an on-the-fly linearization restricted to reachable locations. Although these tools can analyze subsumed affine hybrid systems in principle, bmc algorithms tailored for affine dynamics are generally efficient in comparison. 

%% file: preliminaries.tex
\section{Preliminaries}

\begin{definition}\label{HA_TAH}
    A Hybrid Automaton (HA) \cite{Henzinger-HA, alur1995algorithmic} is a seven tuple $(\V$, $\X$, $\Init$, $\Inv$, $\Flow$, $\Sigma$, $\Trans)$ where:
    \begin{itemize}
        \item $\V = \{v_0,\ldots, v_{\ell}\}$ is a finite set of locations or modes of the hybrid automaton.
        \item $\X = \{x_1,x_2,\dots,x_n\}$ is a finite set of real-valued system variables. The number $n$, the cardinality of this set, represents the \emph{dimension} of the automaton. We write $\dot{\X} = \{\dot{x_1},\dot{x_2},\dots,\dot{x_n}\}$ for the set of dotted variables, which represent the first-order derivatives of the system variables during continuous evolution. We write $\X' = \{x'_1,x'_2,\dots,x'_n\}$ for the set of primed variables that represent the respective modified values after the actuation of a discrete transition.
        
        \item $\Init$, $\Inv$ are functions that map each location in $\V$ to a predicate whose free variables are in $\X$. $\Init$($\ell$) and $\Inv$($\ell$) are called the initial and invariant of the location $\ell$ respectively.  
        \item $\Flow$ is a function that maps each location in $\V$ to a predicate whose free variables are in $\X \cup \dot{\X}$. 
        \item $\Sigma$ is a finite set of automaton labels that include the \emph{stutter} label.
        \item $\Trans$ is a finite set of transitions between locations in $V$. Each transition $\delta \in \Trans$ is a 5-tuple ($\omega$, $v$, $\G$, $\Asgn$, $v'$) where $\omega \in \Sigma$ is the label, $v$ and $v'$ denote the source and destination location, respectively, $\G$ is a predicate whose free variables are from $\X$ and it is called the guard of the transition, and $\Asgn$ is a predicate whose free variables are from $\X \cup \X'$ and it is called the reset assignment of the transition. We require that for each location $v_i \in \V$, there is a \emph{stutter transition} ($stutter$, $v$, $true$, $\bigwedge_{x_i \in \X} x'_i = x_i$, $v$).
    \end{itemize}
\end{definition}

\noindent An affine hybrid automaton (AHA) is a class of hybrid automaton \cite{FLGDCRLRGDM11}. In an AHA, the flow predicate is of the form $\dot{x} = Ax + u$,  where $x \in \mathbb{R}^n$, $A$ is an $n \times n$ real matrix, $u \in \mathcal{U}$ is an input and $\mathcal{U} \subseteq \mathbb{R}^n$ is a compact convex set of inputs. The predicates in the initial, invariant of a location and in the guard of a transition are linear real arithmetic constraints over $\X$, whereas $\Asgn$ of a transition is a linear equality constraint of the form $x_i' = \sum_{i=1}^{n}r_ix_i + c_i$, where $x_i \in \X$ and $r_i$, $c_i \in \mathbb{R}$. 

\noindent \textit{Semantics:} A state of a hybrid automaton is a tuple $(v,x)$, where $v$ is a location and $x$ is a valuation of the variables in $\X$. A state must satisfy the invariant associated with the state's location. A state can change due to a continuous evolution over time, following the flow conditions in $Flow$ of the state's location. This is called a timed transition of states and is represented by the notation $(v,x) \xrightarrow{\tau}(v,y)$, $\tau$ being the dwelling time of the transition. A state of the automaton can also change instantaneously through the transitions of the automaton. Such transitions are called discrete transitions and represented by $(v,x) \xrightarrow{\delta}(v',x')$. A state $(v,x)$ allows for a transition $\delta$ with source location $v$ when $x$ satisfies the guard $\G$ of $\delta$. If an enabled transition is taken, the state changes instantaneously to $(v',x')$ such that $v'$ is the destination location of $\delta$ and $x, x'$ satisfy $\Asgn$ of $\delta$. We denote the graph structure of a hybrid automaton $\H$ by $G_{\H}(\V,E)$ consisting of the locations $\V$ of $\H$ and edges $E$. The edges are taken from the transitions $\delta$ of $\H$ without the guard and reset predicates. The graph structure essentially captures the discrete dynamics of a hybrid system. We now give a formal definition of a run of an automaton that formalizes an execution of the hybrid system:

\begin{definition} [Run of a Hybrid Automaton]
A \emph{run} $\sigma$ of a hybrid automaton is an alternating sequence of timed and discrete transitions:  
$(v_0,x_0) \xrightarrow{\tau_0}(v_0,y_{0})\xrightarrow{\delta_1} (v_1,x_{1}) \xrightarrow{\tau_1} (v_1,y_{1}), \dots,  \xrightarrow{\delta_{k}}(v_k,x_k)\xrightarrow{\tau_{k}} (v_k,y_{k})$
such that (1) $\Init(v_0)[\X := x_0]$ is true (2) $\Inv(v_i)[\X := x_i]$ and $\Inv(v_i)[\X := y_i]$ are both true for $0\leq i \leq k$. (3) In $\G$ and $asgn$ of $\delta_{i+1}$, $\G[\X := y_i]$ and $asgn[\X,\X' := y_i, x_{i+1}]$ are true for $0 \leq i < k$, where $\delta_{i+1} \in \Trans$ is a discrete transition.
The transitions labeled $\tau_i \in \mathbb{R}$ depict timed transitions. 
\end{definition}

\noindent A run is a finite representation of a timed trajectory of the automaton having potentially infinitely many states. The length of $\sigma$ is the number of discrete transitions of the run. We denote $initial(\sigma) = (v_0,x_0)$ and $final(\sigma)=(v_k,y_k)$ to be the initial and final states of the run, respectively. Complex systems are made up of several interacting systems. Such systems can be modeled as a composition of interacting hybrid automata communicating through shared variables and labels. In this work, we assume that there are no shared variables between components. A composition of hybrid automata is defined as follows:   

\begin{definition} [Composition of Hybrid Automata \cite{Henzinger-HA}]\label{def:composed_ha}
    Let $\H_1 = $ $(\V_1$, $\X_1$, $\Init_1$ $\Inv_1$, $\Flow_1$, $\Sigma_1$, $\Trans_1)$ and $\H_2 = $ $(\V_2$, $\X_2$, $\Init_2$, $\Inv_2$, $\Flow_2$, $\Sigma_2$, $\Trans_2)$ be two hybrid automata such that $\X_1 \cap \X_2 = \emptyset$. The composition of $\H_1$ and $\H_2$, denoted by $\H_1 \parallel \H_2$ is a hybrid automaton $\H = (\V$, $\X$, $\Init$, $\Inv$, $\Flow$, $\Sigma$, $\Trans)$ where:
    \begin{itemize}
        \item $\V = \V_1 \times \V_2$; $\X = \X_1 \cup \X_2$
        \item $\Init$, $\Inv$ are functions that maps each location in $\V$ to a predicate over the free variables in $\X$. $\Inv(\langle v_1, v_2\rangle)$ is a predicate given by $\Inv_1(v_1) \land \Inv_2(v_2)$, and $\Init(\langle v_1, v_2\rangle)$ is a predicate given by $\Init_1(v_1) \land \Init_2(v_2)$.
        \item $\Flow$ is a function that maps each location in $\V$ to a predicate whose free variables are in $\X \cup \dot{\X}$. $\Flow(\langle v_1,v_2\rangle)$ is given by the predicate $\Flow_1(v_1) \land \Flow_2(v_2)$.
        \item $\Trans$ is defined as follows:
        \begin{itemize}
            \item For every $(\omega_1$, $v_1$, $\G_1$, $\Asgn_1$, $v_1') \in \Trans_1$ and $(\omega_2$, $v_2$, $\G_2$, $\Asgn_2$, $v_2')$ $\in \Trans_2$, if $\omega_1 = \omega_2$, then $(\omega$, $\langle v_1,v_2\rangle$, $\G_1 \land \G_2$, $\Asgn_1 \land \Asgn_2$, $\langle v_1',v_2'\rangle) \in \Trans$, where $\omega = \omega_1 = \omega_2$ is called the shared or synchronization label.
            \item For every $(\omega_1$, $v_1$, $\G_1$, $\Asgn_1$, $v_1')$ $\in \Trans_1$ such that $\omega_1$ is not a label in any transition in $\Trans_2$, then $(\omega_1$, $\langle v_1,v_2\rangle$, $\G_1$, $\Asgn_1 \land (\bigwedge_{x_i \in \X_2} x'_{i} := x_i)$, $\langle v_1',v_2\rangle)$ $\in \Trans, \forall v_2 \in \V_2$.
            \item For every $(\omega_2$, $v_2$, $\G_2$, $\Asgn_2$, $v_2') \in \Trans_2$ such that $\omega_2$ is not a label in any transition in $\Trans_1$, then $(\omega_2$, $\langle v_1,v_2\rangle$, $\G_2$, $\Asgn_2 \land (\bigwedge_{x_i \in \X_1} x'_{i} := x_i)$, $\langle v_1,v_2'\rangle)$ $\in \Trans, \forall v_1 \in \V_1$.
        \end{itemize}
    \end{itemize}
\end{definition}

\noindent \textit{Composition Semantics:} The composition of $m$ automata $\H_1$, $\H_2$, $\ldots$ , $\H_m$ is denoted by $\H = \H_1 \parallel \H_2 \parallel \ldots \parallel \H_m$ and is defined recursively as $\H1 \parallel \H'$ where $\H' = \H_2 \parallel ... \parallel \H_m$. A state $(\langle v_0,v_1,\ldots,v_m\rangle,x)$ represents the simultaneous states of the components. The variables of each component evolve in parallel following their respective flow conditions. The transitions across components having a shared label must be activated synchronously. In contrast, only one local transition (transition with an unshared label) of any member can be activated at a time. Transitions with a \emph{stutter} label are considered local to a component. The activation of stutter transitions has no effect semantically. When each component is an affine hybrid automaton (AHA), we denote their composition as a compositional affine hybrid automaton (CAHA).

\begin{definition}[A Path in a Composition]
    Given a composition  $\mathcal{H}$ of $m$ hybrid automata, a finite path $\rho$ of length $k$ in the composition is an alternating sequence of locations and transitions: 
        $$\rho = \langle v^0_1,v^0_2,\dots, v^0_m\rangle\xrightarrow{\omega^1}\langle v^1_1,v^1_2,\dots, v^1_m\rangle\xrightarrow{\omega^2}\dots
        \xrightarrow{\omega^k} \langle v^k_1,v^k_2,\dots, v^k_m\rangle$$
 
    \noindent such that $\langle v^i_1,v^i_2,\dots, v^i_m\rangle \in \mathcal{V}$ of $\mathcal{H}$, for $i \in \{0,1, \ldots, k \}$ and $\omega^i$ is the label of a non-stutter transition in $\Trans$ of $\mathcal{H}$ with source as $\langle v^{i-1}_1,v^{i-1}_2,\dots, v^{i-1}_m\rangle$ and destination as $\langle v^{i}_1,v^{i}_2,\dots, v^{i}_m\rangle$, for $i \in \{1,2,\ldots,k\}$. 
\end{definition}
\begin{figure}[h]
    \centering
    \vspace{-15pt}
    \includegraphics[scale=0.55]{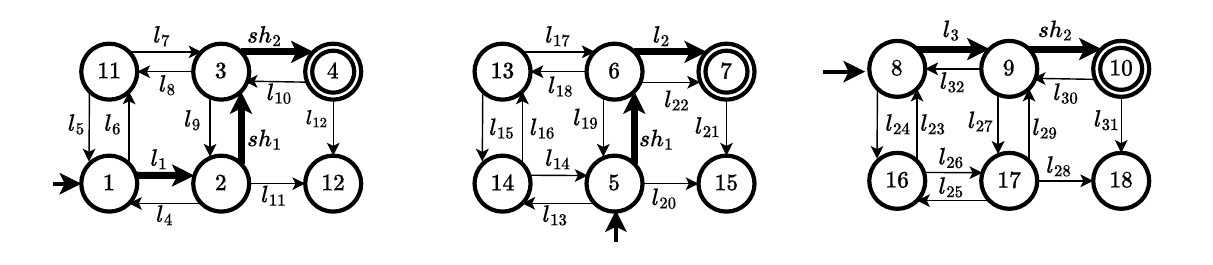}
    \caption{The hybrid automata of three navigation objects (omitting continuous dynamics). The bold transitions demonstrate the compositional path $\langle 1, 5, 8 \rangle \xrightarrow{l_{1}} \langle 2, 5 , 8 \rangle \xrightarrow{sh_{1}} \langle 3, 6, 8 \rangle \xrightarrow{l_{2}} \langle 3, 7, 8 \rangle \xrightarrow{l_{3}} \langle 3, 7, 9\rangle \xrightarrow{sh_2} \langle 4, 7, 10 \rangle$.}
    \label{fig:composed-NAV}
\end{figure}

\begin{wrapfigure}[12]{R}{0.5\textwidth}
    \centering
    \includegraphics[width=0.48\textwidth]{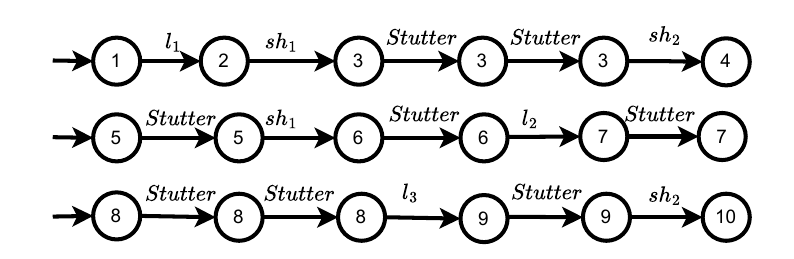}
    \caption{An interleaving compositional path $\langle \rho_1, \rho_2, \rho_3 \rangle$ 
    projected from 
    $\rho = \langle 1, 5, 8\rangle \xrightarrow{l_1} \langle 2, 5, 8 \rangle \xrightarrow{sh_1} 
    \langle 3, 6, 8\rangle \xrightarrow{l_3} \langle 3, 6, 9 \rangle \xrightarrow{l_{2}} 
    \langle 3, 7, 9\rangle \xrightarrow{sh_{2}} \langle 4, 7, 10 \rangle$. }
    \label{fig:interleave_path}
\end{wrapfigure}

As an example to illustrate our algorithm in later sections and to explain some of the definitions relevant to the discussion, we consider a system composed of three navigation objects in the $\mathbb{R}^2$ plane divided into a 2 $\times$ 3 grid of unit cells. Each navigation object is a hybrid system, following affine dynamics of motion. The motion dynamics may be potentially different in each cell. When an object transits from one cell to the other, its dynamics change instantaneously to that of the destination cell. In addition, the objects synchronize their transitions between designated cells. Figure \ref{fig:composed-NAV} shows the locations, transitions, and labels of the hybrid automaton models of the three navigation objects, omitting their continuous dynamics. The first two automata have a shared transition label $sh_1$, while the first and last automata have a shared transition label $sh_2$. The bold transitions in Figure \ref{fig:composed-NAV} show a path that starts from location $\langle 1,5,8\rangle$ and terminates at $\langle4,7,10\rangle$ in the composition.

\begin{definition} [An Interleaving Compositional Path \cite{10.1007/978-3-030-94583-1_23}]
Given a path $\rho = \langle v_1^0,v_2^0,\dots, v_m^0\rangle \xrightarrow{\omega^1} \langle v_1^1,v_2^1,\dots, v_m^1\rangle \xrightarrow{\omega^2} \dots \xrightarrow{\omega^{k}} \langle v_1^k,v_2^k,\dots, v_m^k\rangle$ in $\H = \H_1 \parallel \H_2 \parallel \ldots \parallel \H_m$, the projection of $\rho$ on a member component $\H_i$ is a component path $\rho_i = \langle v_i^0 \rangle \xrightarrow{\omega_i^1} \langle v_i^1 \rangle \xrightarrow{\omega_i^2} \dots \xrightarrow{\omega_i^{k}} \langle v_i^k\rangle$, where $\omega_i^j = \omega^j$ if $\omega^j$ a label of a transition in $\H_i$, otherwise $\omega_i^j$ = stutter. We call the tuple $\langle \rho_1, \rho_2, \dots, \rho_m \rangle$ an interleaving compositional path.
\end{definition}

\begin{wrapfigure}{R}{0.5\textwidth}
        \centering
            \begin{subfigure}[b]{0.5\linewidth}
                \centering
                \includegraphics[scale=0.42]{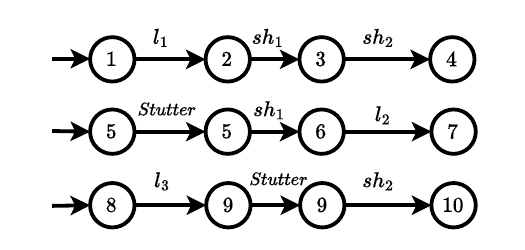}
                \caption{}
                \label{fig:path-intv1}
            \end{subfigure}%
            \begin{subfigure}[b]{0.5\linewidth}
                \centering
                \includegraphics[scale=0.42]{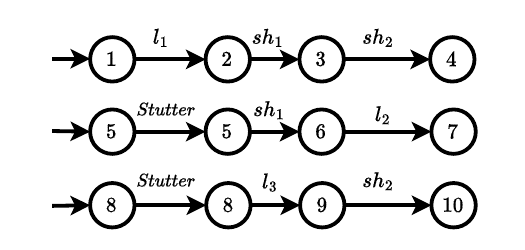}
                \caption{}
                \label{fig:path-step1}
            \end{subfigure}
        \caption{ a) The step compositional path projected from a path in step semantics $\rho_{step} = \langle 1,5,8\rangle \xrightarrow[]{\{l_1,l_3\}}\langle 2,5,9\rangle\xrightarrow{\{sh_1\}}\langle 3,6,9\rangle\xrightarrow[]{\{sh_2,l_2\}}\langle 4,7,10\rangle$. b) Another step compositional path projected from a path in step semantics  $\langle 1,5,8\rangle \xrightarrow[]{\{l_1\}}\langle 2,5,8\rangle\xrightarrow[]{\{sh_1,l_3\}}\langle 3,6,9\rangle\xrightarrow[]{\{sh_2,l_2\}}\langle 4,7,10\rangle$. The local transitions $l_3$ and $l_2$ are activated concurrently with the synchronization transitions $sh_1$ at step 2 and $sh_2$ at step 3, respectively.}
        \label{fig:steppaths}
\end{wrapfigure}
\noindent An interleaving compositional path is a representation of separated local paths of members from a path in the composition. The length of the local paths is kept equal to the length of the compositional path by adding stutter transitions. Figure \ref{fig:interleave_path} shows the interleaving compositional path of the path shown in Figure \ref{fig:composed-NAV}.

\noindent \emph{Step semantics} of composition proposed in \cite{DBLP:conf/forte/BuCLMT10, 10.1007/978-3-030-94583-1_23} allows the simultaneous activation of more than one local transition of members in one \emph{step}. In a synchronization transition, the non-syncing members may simultaneously activate their local transitions, while the syncing members must activate the shared transitions in synchrony. The tuple of paths of components projected from a compositional path in step semantics is called a \emph{step compositional path}.

\noindent Figure \ref{fig:steppaths} depicts a couple of step compositional paths of 3 steps between $v^{init}$ and  $v^{unsafe}$ of the composition in Figure \ref{fig:composed-NAV}. In Fig. \ref{fig:path-intv1}, the non-stutter transitions activated together in a step of a step compositional path $\rho_{step}$ are represented as a set of their labels. For instance, the two local transitions $l_1$ and $l_3$ are activated simultaneously in step 1. Similarly, a local transition $l_2$ is triggered simultaneously with the synchronization transition $sh_2$ in step 3 in Fig. \ref{fig:path-step1}. The path shown in Figure \ref{fig:composed-NAV} is succinctly represented by either of the step compositional paths. We define the length of a step compositional path as the number of steps it contains. For example, the length of the step compositional paths in Figure \ref{fig:steppaths} is 3.

\begin{wrapfigure}[8]{R}{0.4\textwidth} 
    \centering
    \vspace{-25pt}
    \includegraphics[width=0.9\linewidth]{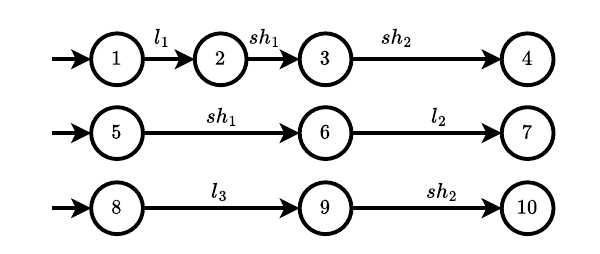}
    \caption{The shallow compositional path $\rho_{shallow}$ is obtained by removing stutter transitions from the path $\rho_{step}$.}
    \label{fig:shallowpath}
\end{wrapfigure}

\noindent A stutter-free variant of a step composition path $\rho_{step}$ is called a \emph{shallow} compositional path $\rho_{shallow}$ \cite{10.1007/978-3-030-94583-1_23}. A \emph{shallow} compositional path is obtained by removing the stutter transitions from a step compositional path. For example, Figure \ref{fig:shallowpath} shows the shallow compositional path obtained from the step compositional paths defined in Figure \ref{fig:steppaths}.

\begin{definition} [Bounded Safety Specification]\label{defn:bounded_safety_spec}
Given a composition of $m$ hybrid automata $\H=\H_1 \parallel \H_2 \parallel \ldots \parallel\H_m$, an initial configuration $\I=(v^{init},\C^{init})$ where $\C^{init} \subseteq \Inv(v^{init})$, an unsafe configuration $\F = (v^{unsafe}$, $~\C^{unsafe})$ where $v^{unsafe} \in \V$ and $\C^{unsafe} \subseteq \Inv(v^{unsafe})$ and an analysis bound $k$, the composition $\H$ is said to be safe if there exists no run $\sigma$ of $\H$ such that the length of $\sigma$ is $\leq k$, $initial(\sigma) \in \I$ and $final(\sigma) \in \F$.
\end{definition}

\noindent In this paper, we address the problem of deciding the bounded safety of the composition of $m$ affine hybrid automata.

%% file: methodology.tex
\section{Methodology}

We now present the CEGAR-based bounded reachability analysis procedure, which combines SAT-based compositional path enumeration with symbolic reachability analysis. In this context, a compositional path refers to an abstract counterexample identified in the discrete abstraction, while the symbolic reachability analysis serves as the refinement phase to validate this abstract counterexample. Figure \ref{fig:2} describes the flow of our proposed algorithm, which takes as input a bounded safety specification (see Defn \ref{defn:bounded_safety_spec}) and $m$ affine hybrid automata $\H_1, \H_2, \dots, \H_m$. The algorithm terminates with the following two outcomes: a) the instance is safe, which means the forbidden states are not reachable from the initial states for the given bound of analysis. b) The safety of the instance is unknown. This is concluded when the reachable states intersect with the forbidden states within the analysis bound. In this case, it remains undecided whether the reachability is due to the over-approximate reachable set computation or due to actual unsafety. The first step in our procedure is to retrieve the step compositional paths of bounded length from the initial to the unsafe location. These paths are obtained by constructing a propositional logic formula from the graph structure of the component automaton. The satisfiability of the formula results in finding a step compositional path from the satisfying assignment. Then, a step compositional path is translated to the corresponding shallow compositional path. The reason for converting from a path in step semantics to shallow semantics is discussed in \cite{10.1007/978-3-030-94583-1_23}, and we shall also briefly discuss the motivation in Section \ref{postC_postD}. The next step is symbolic state-space computation of the shallow compositional path to decide the reachability of the unsafe location due to continuous dynamics. When unreachable, the negation of the path is conjoined with the formula to exclude it in the retrieval of the next path. While exploring the state space of a path that turns out to be safe, additional safe paths are inferred and eliminated from further enumeration and exploration, boosting efficiency. 

\begin{wrapfigure}[21]{r}{0.56\textwidth} 
    \centering
    \includegraphics[scale=0.36]{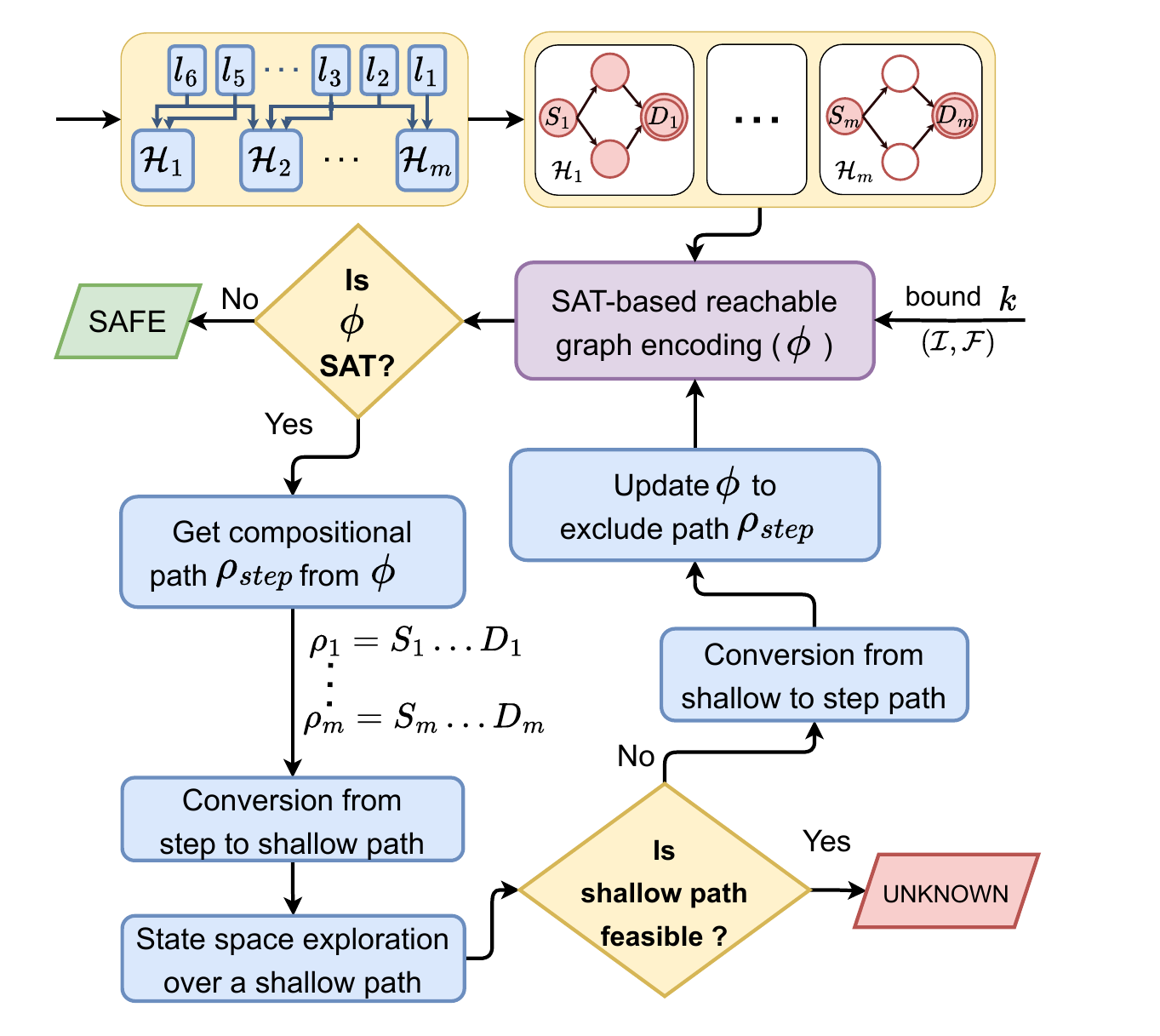}
    \caption{Flow of our proposed algorithm. The arrow indicates the workflow between the execution steps.}
    \label{fig:2}
\end{wrapfigure}

\noindent In addition, an optimization technique has been introduced to improve computational efficiency by leveraging memoization of the computed reachable states during state-space exploration. When all possible paths turn out to be safe, the algorithm terminates by declaring the composition safe. On the contrary, when at least one path is found unsafe, the algorithm concludes the safety of the instance to be unknown, due to the over-approximation in state-space computation. The subsequent sections discuss the details.

\subsection{Motivating Example}\label{motivating_exam}

\begin{table}[]
\caption{The compositional paths of length $5$ represented by step compositional paths shown in Fig. \ref{fig:steppaths} }
\footnotesize
\label{tab:inpaths}
\begin{equation}
\begin{aligned}
    \text{a) } & \langle 1, 5, 8 \rangle \xrightarrow{l_{1}} \langle 2, 5 , 8 \rangle \xrightarrow{sh_{1}} \langle 3, 6, 8 \rangle \xrightarrow{l_{2}} \langle 3, 7, 8 \rangle \xrightarrow{l_{3}} \langle 3, 7, 9\rangle \xrightarrow{sh_2} \langle 4, 7, 10 \rangle \\
    \text{b) } & \langle 1, 5, 8 \rangle \xrightarrow{l_{1}} \langle 2, 5 , 8 \rangle \xrightarrow{sh_{1}} \langle 3, 6, 8 \rangle \xrightarrow{l_{3}} \langle 3, 6, 9 \rangle \xrightarrow{l_{2}} \langle 3, 7, 9\rangle \xrightarrow{sh_2} \langle 4, 7, 10 \rangle \\
    \text{c) } & \langle 1, 5, 8 \rangle \xrightarrow{l_{1}} \langle 2, 5 , 8 \rangle \xrightarrow{l_{3}} \langle 2, 5, 9 \rangle \xrightarrow{sh_{1}} \langle 3, 6, 9 \rangle \xrightarrow{l_{2}} \langle 3, 7, 9\rangle \xrightarrow{sh_2} \langle 4, 7, 10 \rangle \\
    \text{d) } & \langle 1, 5, 8 \rangle \xrightarrow{l_{3}} \langle 1, 5 , 9 \rangle \xrightarrow{l_{1}} \langle 2, 5, 9 \rangle \xrightarrow{sh_{1}} \langle 3, 6, 9 \rangle \xrightarrow{l_{2}} \langle 3, 7, 9\rangle \xrightarrow{sh_2} \langle 4, 7, 10 \rangle \\
    \text{e) } & \langle 1, 5, 8 \rangle \xrightarrow{l_{1}} \langle 2, 5 , 8 \rangle \xrightarrow{l_{3}} \langle 2, 5, 9 \rangle \xrightarrow{sh_{1}} \langle 3, 6, 9 \rangle \xrightarrow{sh_{2}} \langle 4, 6, 10\rangle \xrightarrow{l_2} \langle 4, 7, 10 \rangle \\
    \text{f) } & \langle 1, 5, 8 \rangle \xrightarrow{l_{1}} \langle 2, 5 , 8 \rangle \xrightarrow{sh_{1}} \langle 3, 6, 8 \rangle \xrightarrow{l_{3}} \langle 3, 6, 9 \rangle \xrightarrow{sh_{2}} \langle 4, 6, 10\rangle \xrightarrow{l_2} \langle 4, 7, 10 \rangle \\
    \text{g) } & \langle 1, 5, 8 \rangle \xrightarrow{l_{3}} \langle 1, 5 , 9 \rangle \xrightarrow{l_{1}} \langle 2, 5, 9 \rangle \xrightarrow{sh_{1}} \langle 3, 6, 9 \rangle \xrightarrow{sh_{2}} \langle 4, 6, 10\rangle \xrightarrow{l_2} \langle 4, 7, 10 \rangle
\end{aligned}
\end{equation}
\end{table}
\noindent Consider the bounded safety verification problem of the compositional system of Figure \ref{fig:composed-NAV}, given $v^{init}=\langle 1,5,8\rangle$ and $v^{unsafe}=\langle 4,7,10\rangle$ for some initial and unsafe configuration, respectively, and for a bound of analysis, say $5$. Bounded reachability analysis algorithms that explicitly compute the product of the member automata and then explore the state-space of the composition spend considerable computational resources in terms of time and memory in product construction. For example, the product of the three components results in a composition with $216$ locations and $1164$ transitions. Secondly, algorithms that explore the state-space of the product automaton using uninformed search algorithms such as BFS/DFS may spend computational resources in irrelevant regions of the state-space (exploring regions of the automaton that can never lead to an unsafe state). In contrast, our proposed algorithm does not explicitly construct the product of the member automata. Moreover, a CEGAR-based state-space exploration guides the search along relevant regions of the state-space. A detailed discussion on the advantages of CEGAR-based state-space exploration in affine monolithic hybrid automata can be found in \cite{10.1145/3567425}. The first step of our algorithm is to enumerate all compositional paths in the step semantics that start from $v^{init}$ and end at $v^{unsafe}$ and map to compositional paths in the interleaving semantics of length less than or equal to the given bound. 

\begin{wrapfigure}{R}{0.5\textwidth}
        \centering
        \vspace{-15pt}
            \begin{subfigure}[b]{0.5\linewidth}
                \centering
                \includegraphics[scale=0.42]{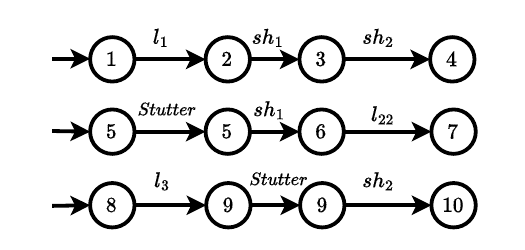}
                \caption{}
                \label{fig:path-intv}
            \end{subfigure}%
            \begin{subfigure}[b]{0.5\linewidth}
                \centering
                \includegraphics[scale=0.42]{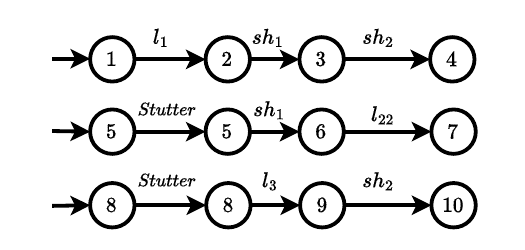}
                \caption{}
                \label{fig:path-step}
            \end{subfigure}
        \caption{a) The step compositional path obtained as a projection of a path in step semantics $\rho_{step'} = \langle 1,5,8\rangle \xrightarrow[]{\{l_1,l_3\}}\langle 2,5,9\rangle\xrightarrow{\{sh_1\}}\langle 3,6,9\rangle\xrightarrow[]{\{sh_2,l_{22}\}}\langle 4,7,10\rangle$. b) Another step compositional path obtained as a projection of a path in step semantics  $\langle 1,5,8\rangle \xrightarrow[]{\{l_1\}}\langle 2,5,8\rangle\xrightarrow[]{\{sh_1,l_3\}}\langle 3,6,9\rangle\xrightarrow[]{\{sh_2,l_{22}\}}\langle 4,7,10\rangle$.}
        \label{fig:steppaths2}
    \end{wrapfigure}

\noindent In our example of Figure \ref{fig:composed-NAV}, there are $14$ compositional paths starting from $v^{init}$ and ending at $v^{unsafe}$ of length bounded by 5. Table \ref{tab:inpaths} shows 7 of these paths, and the rest of the 7 paths can be obtained by replacing the transitions of the label $l_2$ with the transitions of the label $l_{22}$. The path enumeration in our procedure finds four step-compositional paths shown in Fig. \ref{fig:steppaths} and \ref{fig:steppaths2}. Among these, the two step compositional paths shown in Fig. \ref{fig:steppaths} both capture the seven compositional paths depicted in Table \ref{tab:inpaths}. Similarly, the remaining step compositional paths shown in Fig. \ref{fig:steppaths2} represent another set of seven compositional paths, where transition $l_{2}$ is replaced by transition $l_{22}$. Therefore, the enumeration of two step compositional paths suffices to capture all 14 compositional paths, improving efficiency. 

\begin{wrapfigure}{r}{0.35\linewidth} 
    \centering
    \includegraphics[scale=0.5]{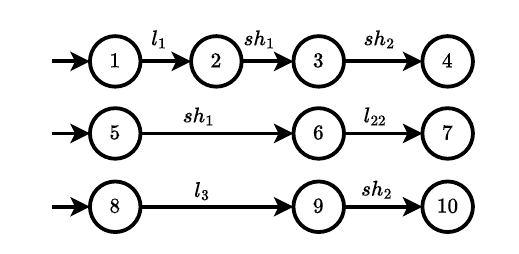}
    \caption{Another shallow compositional path $\rho_{shallow'}$ obtained from a step compositional path $\rho_{step'}$ after removing the stutter transitions.}
    \label{fig:shallowpath2}
\end{wrapfigure}
\noindent The redundant step compositional paths are pruned by incorporating optimization techniques in the path enumeration process. Consequently, these paths are converted to paths in shallow compositional semantics. Specifically, the step compositional paths shown in Fig. \ref{fig:steppaths} map to a single shallow compositional path as shown in Fig. \ref{fig:shallowpath}, while the paths shown in Fig. \ref{fig:steppaths2} map to another shallow compositional path shown in Fig. \ref{fig:shallowpath2}. Note that the analysis of 14 compositional paths boils down to the analysis of a couple of shallow compositional paths. Finally, for each shallow compositional path, the set of reachable states arising from the continuous dynamics and consistent with the path is computed, with the objective of verifying the bounded safety. To compute the reachable states, a shallow path is unfolded to get a computation tree, a section of the product automaton, capturing the various inter-leavings of transitions in the path. For example, the computation tree of the shallow compositional path $\rho_{shallow}$ is shown in Fig. \ref{fig:composed-state-space}. The computation tree for the other shallow path $\rho_{shallow'}$ will be similar except that the transitions of the label $l_2$ will be replaced by the transitions of the label $l_{22}$. The computation tree is explored in a depth-first manner with successor symbolic state computation operators, which will be discussed in Section \ref{postC_postD}. Furthermore, observe that the computation trees of $\rho_{shallow}$ and $\rho_{shallow'}$ share most of the locations and transitions. The result of exploring the computation tree of one path is memoized and reused while exploring the computation tree of the other path, enabling further efficiency. Due to the path-centric nature of exploration, the algorithm computes the symbolic state successor of only $10$ locations of the computation tree shown in Figure \ref{fig:composed-state-space}, out of the 216 locations of the complete product automaton. 

\begin{wrapfigure}[16]{R}{0.25\linewidth} 
     \centering
    \includegraphics[scale=0.5]{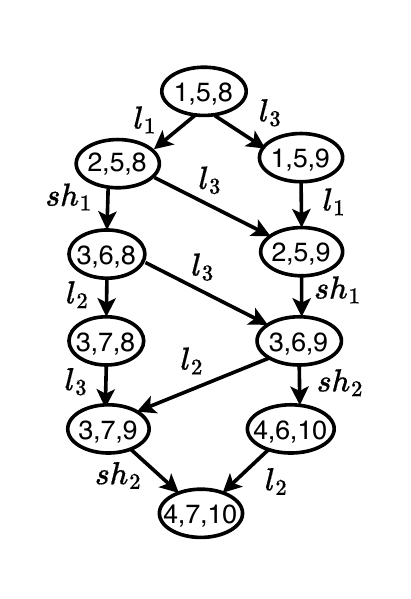}
    \caption{The computation tree unfolded from the shallow compositional path in Figure \ref{fig:shallowpath}.}
    \label{fig:composed-state-space}
\end{wrapfigure}

\subsection{Retrieving Step Compositional Paths}
Step semantics of composition permit multiple transitions to occur within a single step, and therefore, a step compositional path can represent various interleaving compositional paths. Analysis of a step compositional path involves examining all the interleaving compositional paths that it represents. As a result, enumerating the step compositional paths significantly accelerates the compositional path enumeration procedure to verify safety requirements. We first show that all interleaving compositional paths represented by a step compositional path have the same length.

\begin{proposition}\label{prop1}
Let $S_1, S_2,\dots, S_l$ be the set of labels of non-stutter transitions in each step of a step compositional path $\rho_{step}$ of length $l$. The interleaving path(s) succinctly represented by $\rho_{step}$ have length $ \sum_{i=1}^{l}|S_i|$.  
\end{proposition}

\begin{proof}
 In interleaving semantics, no two transitions labeled by $S_i$ in step $i$ of a path $\rho_{step}$ can occur simultaneously. Therefore, the length of any interleaving path of $\rho_{step}$ will be of length $ \sum_{i=1}^{l}|S_i|$. The distinct interleavings of transitions in $\rho_{step}$ give the various interleaving compositional paths that it represents; all will have length $ \sum_{i=1}^{l}|S_i|$.
\end{proof}

\begin{proposition}\label{prop2}
    For every interleaving path $\rho$ of length $l$ in a composition $\H$, there exists a step path $\rho'$ of $\H$ of length at most $l$ that represents $\rho$.
\end{proposition}

\begin{proof}
    We construct a step compositional path $\rho'$ that represents $\rho$ by grouping consecutive local (unshared) transitions of $\rho$ into a single step transition. This yields a finite sequence of steps $S_1, S_2, \dots, S_l$, where each $S_i$ is a set of local non-stutter transitions that could be executed in a step. Since one or more local transitions are merged in a single step, the length of $\rho'$ is bounded by the length of $\rho$.
\end{proof}

\noindent We now present an encoding of a propositional logic formula $\phi^k_{\H}$ from graph structures $G_{\H_1}$($\V_1$, $E_1$), $G_{H_2}(\V_2,E_2)$, $\dots$, $G_{\H_m}(\V_m,E_m)$ associated with the component AHAs. The encoding is such that the satisfiability of $\phi^k_{\H}$ implies the existence of a step compositional path $\langle \rho_1, \rho_2, \dots, \rho_m \rangle$ starting from $v^{init}$ and ending at $v^{unsafe}$ respectively. A satisfiable assignment of the variables in $\phi^k_{\H}$ can be decoded back to obtain a step compositional path $\rho_{step}$. We describe the encoding of $\phi^k_{\H}$ using five primary constraints, namely \emph{initial constraint}, \emph{exclusivity constraint}, \emph{transition constraint}, \emph{shared label constraint} and \emph{destination constraint}, which are denoted as $\phi_{init}, \phi_{excl}, \phi_{trans}$, $\phi_{shared}$ and $\phi_{dest}$ respectively. The encoding scheme is motivated by \cite{10.1007/978-3-030-94583-1_23} with our proposed optimizations. We illustrate below the constraints of a component, say $\H_c$. These constraints are encoded using boolean variables $t_c^j$, $stutter_c^j$ and $sh\_label_{\omega}^j$ for $1 \leq j \leq k$, $1 \leq c \leq m $ and $\omega \in \Sigma_{sh}$, $\Sigma_{sh}$ being the set of shared labels. The truth of $t_c^j$ represents that the transition $t_c$ of $\H_c$ is enabled in step $j$ of the path. Similarly, the truth of $stutter_c^j$ represents that the transition in step $j$ of the component $\H_c$ is a stutter transition. The truth of $sh\_label_{\omega}^j$ implies that transitions with a shared label $\omega$ have been taken in the compositional path at step $j$. 

\noindent \textbf{Initial constraint.} The \emph{initial constraint} for a component $\H_c$ is encoded as:
\begin{equation*}
\phi_{init}(\H_c) = \bigvee_{t_c \in T_{out}(v_c^{init})} t_c^1
\end{equation*}
This constraint enforces that a path of $\H_c$ in the composition must start with one of the outgoing transitions of the initial location. $T_{out}$ of a location returns the set of outgoing transitions from that location, including a stutter transition.

\noindent \textbf{Transition constraint.} The \emph{transition constraint} ensures that any path respects the transition relations. If a transition $t_i$ is taken at step $j$, then the next transition in the path at step $j+1$ must be one of the outgoing transitions of the destination location of $t_i$. We use two functions, $L(t_c)$ and $Reach_c(j)$, to generate the constraint. $L(t_c)$ returns the destination location of a transition $t_c \in E_c$. When $t_c$ is a stutter transition, $L(t_c)$ returns the destination location of the first preceding non-stutter transition in the path, or it returns $v_c^{init}$ when all preceding transitions are stutter. $Reach_c(j)$ returns the set of transitions which are $j$ steps away from $v_c^{init}$ in $G_{\H_c}(\V_c,E_c)$. Hence, the non-stutter transition in $step$ $j$ of a path must be one in $Reach_c(j)$. The use of this function is to shorten the formula size since $|Reach_c(j)| \leq |E_c|$.

\begin{equation}
\phi_{\text{trans}}(\mathcal{H}_c, j) = 
\bigwedge_{t_c \in \text{Reach}_c(j) \cup \{\text{stutter}\}} 
\Biggl(
t_c^j \implies \\
\bigvee_{t_{c'} \in T_{\text{out}}(L(t_c))} 
t_{c'}^{j+1}
\Biggr)
\end{equation}

\noindent \textbf{Exclusivity constraint.} The \emph{exclusivity constraint} states that if a transition is taken at a step $j$, then no other transition can be taken at the same step.
\begin{equation}
\phi_{\text{excl}}(\mathcal{H}_c, j) = 
\bigwedge_{t_c \in \text{Reach}_c(j)} 
\Biggl(
t_c^j \implies \neg \biggl( \\
\bigvee_{t_{c'} \in (\text{Reach}_c(j) - \{t_c\}}) 
t_{c'}^j
\biggr)
\Biggr)
\end{equation}

\noindent \textbf{Destination constraint.} This encodes that at depth $j$, the transition must be an incoming transition to the forbidden location. In the encoding, $T_{in}(v)$ is a function that returns the incoming transitions of a location $v$. When no incoming transition to the unsafe location is reachable in step $j$, the destination constraint becomes an empty clause, which is interpreted as false.

$$\phi_{dest}(\H_c, j) = \bigvee_{t_{c}\in T_{in}\big((v_c^{unsafe}) \cap Reach(j)\big)}{t_{c}^j} $$

\noindent Now, combining the above constraints, we get the following encoding of a $k$-bounded unfolding of $G_{\H_c}(\V_c,E_c)$:

\begin{equation}
\phi^k_{\mathcal{H}_c} = 
\phi_{\text{init}}(\mathcal{H}_c) \land 
\bigwedge_{1 \leq j \leq k-1} \phi_{\text{trans}}(\mathcal{H}_c, j) \land \\
\bigwedge_{1 \leq j \leq k} \phi_{\text{excl}}(\mathcal{H}_c, j) \land 
\bigvee_{1 \leq j \leq k} \phi_{\text{dest}}(\mathcal{H}_c, j)
\end{equation}

\noindent We now present the encoding to combine the constraints of each member in the composition. We first present the shared label constraint:

\noindent \textbf{Shared label constraint.} To synchronize shared transitions across components, we have the following constraint to track an occurrence of a shared transition in any step $j$ using the boolean variable $sh\_label_{\omega}^j$. The function $label(t)$ returns the label of a transition $t$.
\begin{equation}
\phi_{\text{shared}}(\mathcal{H}_c, j) = 
\bigwedge_{\{t_c \in E_c \,|\, \text{label}(t_c) \in \Sigma_{\text{sh}} \}}
\Big( t_c^j \iff \\
\text{sh\_label}_{\text{label}(t_c)}^j \Big)
\end{equation}

\noindent \textbf{Exclusive shared constraint.} In addition, we have to forbid transitions with distinct shared labels to be taken at the same step. The \emph{exclusive shared} constraint makes that happen:

\begin{equation}
\phi_{\text{excl\_shared}}(j) = 
\bigwedge_{\omega \in \Sigma_{\text{sh}}} 
\Big( 
\text{sh\_label}_{\omega}^j \implies \\
\bigwedge_{\omega' \in (\Sigma_{\text{sh}} - \{\omega\})} 
\lnot \text{sh\_label}_{\omega'}^j 
\Big)
\end{equation}

\noindent \textbf{Synchronization constraint.} The primary concern is to respect the synchronization semantics for which we introduce the \emph{synchronization constraint} denoted by $\phi_{sync}$. This constraint ensures that a shared transition must be true simultaneously in all its members at any step in a compositional path. The synchronization constraint is given as:
\begin{equation}
\phi_{\text{sync}}^j = 
\bigwedge_{\omega \in \Sigma_{\text{sh}}} 
\Biggl( 
\text{sh\_label}_{\omega}^j \implies 
\bigwedge_{1 \leq c \leq m} \\
\Big( 
\bigvee_{t_c \in \{t_c \in E_c \,|\, \text{label}(t_c) = \omega\}} t_c^j 
\Big) 
\Biggr)
\end{equation}

\noindent Joining the member encoding together with shared label constraint, exclusive shared constraint, and synchronization constraint, the $k$ bounded unfolding of the graph of a compositional hybrid automaton $\mathcal{H} = \H_1 \parallel \H_2 \parallel \ldots \parallel \H_m$ is represented as:
\begin{equation}\label{eq:2}
    \phi^k_{\mathcal{H}} = {\bigwedge_{1 \leq c \leq m} \Big( {\phi^k_{\H_c}}} \land \bigwedge_{1 \leq j \leq k}  \phi_{shared}(\H_c,j)  \Big) \land \\
    \bigwedge_{1\leq j \leq k} \phi_{excl\_shared}(j) \land \bigwedge_{1\leq j \leq k}{\phi_{sync}^j}  
\end{equation}
\noindent There are a maximum of $k.|E_c|$ transition variables, $k$ stutter variables and $|E_c|$ shared label variables in $\phi^k_{\H_c}$. In $\phi^k_{\mathcal{H}}$, the number of variables is $O(k.m.|E|)$, where $|E| = \sum_{c=1}^m |E_c|$. The length of $\phi^k_{\mathcal{H}}$ is $O (k. \sum_{c=1}^m |E_c|^2)$.

\noindent Algorithm \ref{algo:path_enum} shows the enumeration of the step compositional paths, which retrieves step compositional paths of length $\leq k$. By Proposition \ref{prop1}, since all interleaving paths represented by a step compositional path have the same length, the algorithm safely discards paths whose corresponding interleaving paths are longer than $k$, for bounded reachability analysis.

\begin{algorithm}
\caption{Step path enumeration.}
    \scriptsize
    \begin{algorithmic}[1]
        \Require $m$ affine hybrid automata $\H_1, \H_2, \dots, \H_m$ along with $\I$ and $\F$, and bound of analysis $k$.
        \Ensure Unknown or Safe.
        \For{$l = 1$ to $k$}
            \State Construct the formula $\phi^l_{\mathcal{H}}$.\Comment{defined in equation \ref{eq:2}.}
            \State valid\_path = $False$
            \State path\_found = $False$
                \While{$\phi^l_{\mathcal{H}}$ is SAT}
                    \State Find a step path $P = S_1,S_2,\dots,S_l$ 
                    \State path\_found = $True$
                    \If{$\sum_{i=1}^{l}|S_i| \leq k$}
                        \State valid\_path = $True$
                        \State infeasible, $\rho'_{inf} = RA(P, \I,\F)$\Comment{See section \ref{postC_postD}.}
                        \If{infeasible is $False$}
                            \State return Unknown.
                        \Else
                            \State Update $\phi^l_{\mathcal{H}}$ by $\rho'_{inf}$.\Comment{See section \ref{prune_section}}
                        \EndIf
                    \Else
                        \State Update $\phi^l_{\mathcal{H}}$ by $P$.\Comment{See section \ref{prune_section}}
                    \EndIf
                \EndWhile
            \If{(valid\_path is $False$ and path\_found is $True$)}
               \State return Safe. 
            \EndIf
        \EndFor
       \State \Return Safe
    \end{algorithmic}
    \label{algo:path_enum}
\end{algorithm}

\noindent In order to later prove the correctness of the proposed bounded reachability analysis algorithm, we now show that no interleaving path of length bounded by $k$ is skipped in the path enumeration.

\begin{claim}\label{claim_1}
The enumerated step compositional paths in Algorithm \ref{algo:path_enum} represent all interleaving compositional paths of length bounded by $k$.
\end{claim}
\begin{proof}
    Algorithm \ref{algo:path_enum} enumerates all valid step compositional paths of length up to $k$. For all interleaving paths $\rho$ of length $l \leq k$, proposition \ref{prop2} states that there exists a step path $\rho'$ of length bounded by $l$ that represents $\rho$. Since $l \leq k$, Algorithm \ref{algo:path_enum} will enumerate such a $\rho'$. In addition, Algorithm \ref{algo:path_enum} discards the step paths that correspond to interleaving paths larger than $k$. 
\end{proof}

\subsection{Conversion to Shallow Compositional Path and Reachability Analysis}\label{postC_postD}

The next step in the algorithm is a reachability analysis of the enumerated paths in Algorithm \ref{algo:path_enum}. Before reachability analysis of a step path, it is converted into a shallow path. We now discuss the motivation for this conversion.  

\paragraph{Conversion to Shallow Path} Many step paths represent the same set of interleaving paths, due to the presence of stutter transitions. For example, the two step paths shown in Fig. \ref{fig:steppaths} represent the same set of interleaving paths shown in Table \ref{tab:inpaths}. This leads to redundancy in computation due to repeated reachability analysis of similar step paths. To address this issue, a step compositional path is converted into a stutter-free shallow path, and reachability analysis is performed on the resulting shallow path only once. Note that step paths that are stutter variants map to a unique shallow path. For example, the step compositional path depicted in Figure \ref{fig:path-step1} can be obtained by altering the order of labels $l_{3}$ and $Stutter$ in Fig~\ref{fig:path-intv1} and vice-versa. These two step compositional paths are shallowly consistent, meaning that their corresponding stutter-free shallow path is the one depicted in Fig~\ref{fig:shallowpath}. Thus, converting a step path into a shallow path and consequently performing a reachability analysis on the converted shallow path saves redundant reachability analysis. 

\paragraph{Reachability Analysis}
An interleaving path in a composition can be interpreted as an abstract representation of all runs of the composition that follow the same sequence of locations and transitions as in the path. We formalize this relation between a path and a run of a composition as follows.  

\begin{definition} A run $\sigma$ of the automaton $\H$ is said to be \emph{consistent} with an interleaving compositional path $\rho$ of $\H$ if and only if the sequence of locations and transitions observed in $\sigma$ (with consecutive locations in a timed transition treated as a single location) is exactly the same as those appearing in $\rho$. We denote this relation as $\sigma \sim \rho$.
\end{definition}

\noindent The reachability relation between a state $(v,z)$ and a run $\sigma$ of $\H$ is given below. 

\begin{definition}
A state $(v,z)$ of $\H$ is said to be reachable by a run $\sigma$ of $\H$ if and only if there exists a timed transition $(v,x) \xrightarrow{\tau} (v,y)$ in $\sigma$ and a time $t \in [0,\tau]$ such that $(v,x) \xrightarrow{t} (v,z)$ is also a timed transition of $\H$. 
\end{definition}

\noindent The above definition essentially relates the states on the timed trajectory of a run to be reachable by the run. We now define the set of reachable states for an interleaving compositional path $\rho$.

\begin{definition} \cite{10.1145/3567425}
Given an interleaving compositional path $\rho$ of the automaton $\H$, $\mathcal{R}_{\rho}$ is the set of states reachable by all \emph{runs} $\sigma$ of $\mathcal{H}$ such that $\sigma \sim \rho$.  
\end{definition}

\noindent Recall from the discussion of Section~\ref{motivating_exam} that a shallow compositional path $\rho_{shallow}$ represents one or more interleaving compositional paths. We use the notation $\mathcal{R}_{shallow}$ to represent the union of $\mathcal{R}_{\rho}$ on all interleaving paths $\rho$ represented by $\rho_{\rho_{shallow}}$. To compute the reachable states of $\rho_{shallow}$, we systematically explore the computation tree of $\rho_{shallow}$ in a breadth-first manner and compute the reachable states for all interleaving paths that it represents. The union of these reachable states constitutes the set of reachable states of $\rho_{shallow}$.

\noindent Once $\rho_{shallow}$ is obtained from $\rho_{step}$, the next step of our algorithm is to compute the reachable states $R_{\rho_{shallow}}$ of $\rho_{shallow}$. The reachability analysis decides the infeasibility of $\rho_{shallow}$ defined as follows:

\begin{definition}
A path $\rho_{shallow}$ of a compositional hybrid automaton $\mathcal{H}$ is called \emph{infeasible} if and only if for all interleaving paths $\rho$ represented by $\rho_{shallow}$, there is no run $\sigma$ of $\mathcal{H}$ such that $\sigma \sim \rho$.
\end{definition}

\begin{wrapfigure}{r}{0.35\textwidth}
    \centering
    \includegraphics[scale=0.4]{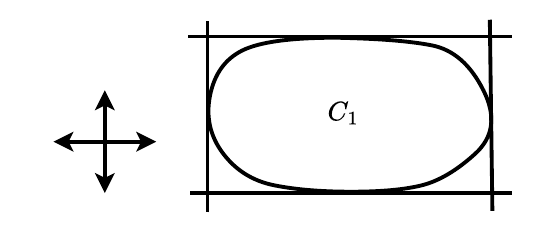}
    \caption{The rectangular region is the template polyhedron approximation of the convex set $C_1$.}
    \label{fig:tem_polyhedron}
\end{wrapfigure}

\paragraph{Support Function Representation}
In order to compute reachable states, we perform a symbolic state-space exploration similar to \cite{10.1145/3567425}, now over a compositional path. The operations are on a \emph{symbolic state} which is a finite representation of potentially infinitely many states of a hybrid automaton. A \emph{symbolic state} of the composition $\H = \H_1 \parallel \H_2 \parallel ...\parallel \H_m$ is given by ($\textbf{V},\textbf{C}$), where $\textbf{V} \in \V_1 \times \V_2 \ldots \times \V_m$ and $\textbf{C} \subseteq \Inv(\textbf{V}) \subseteq  \mathbb{R}^{(|\X_1| + |\X_2|+ \dots + |\X_m|)} $. The continuous set $\textbf{C}$ is stored as a union of convex sets represented by their \emph{support functions}. The support function of a convex set $\Omega \subseteq \mathcal{R}^n$ is defined as $h_{\Omega}:\mathbb{R}^n \to \mathbb{R}$, where $h_{\Omega}(\ell) = \underset{x \in \Omega}{\text{max }} \ell \cdot x$. The algorithm uses the $postC$ operator to compute all timed successors (say $\textbf{V},\textbf{C}^{flow}$) of a given symbolic state, where $\textbf{C}^{flow}$ is a union of bounded convex sets $\Omega_0, \Omega_1,\dots,\Omega_k$ such that each $\Omega_i \in \Inv(\textbf{V})$. Each $\Omega_i$ represents an over-approximation of reachable states within the time interval [$i\Delta t, i\Delta t + \Delta t$], where $i \leq \lceil \frac{T}{\Delta t} \rceil - 1$, $T$ is the time horizon, and $\Delta t$ denotes the time step. The computation of $\Omega_i$ is illustrated in \cite{10.1145/3567425}. The intersection of continuous sets is a frequent operation in reachability analysis, and it is computationally inefficient for convex sets represented by their support functions.  Intersection operation between convex sets is cheap when represented as \emph{template polyhedra} instead. Therefore, each $\Omega_i$ in $\textbf{C}^{flow}$ is converted to a template polyhedron approximation, with a precision trade-off.

\begin{definition}[Template Polyhedron] \cite{FLGDCRLRGDM11}
    Given a finite set of vectors $D= \{\ell_1,\ell_2,\dots,\ell_d\}$ in $\mathbb{R}^n$ known as template directions, a template polyhedron $P_{D} \subseteq \mathbb{R}^n$ is a polyhedron for which there exist coefficients $b_1,b_2,\dots,b_d \in \mathbb{R}^n$ such that
    $$P_{D} = \{x \in \mathbb{R}^n|\bigwedge_{\ell_i\in D}\ell_i.x\leq b_i\}$$
\end{definition}

\noindent Given template directions $D$, a template polyhedron approximation of a convex set $\Omega$ is defined by 

$$P_{D}(\Omega) = \{x \in \mathbb{R}^n|\bigwedge_{\ell_i\in D}\ell_i.x\leq h_{\Omega}(\ell_i)\}$$

\noindent Figure \ref{fig:tem_polyhedron} shows the template polyhedron approximation of a convex set $C_1 \subset \mathbb{R}^2$ along the \emph{axis} directions as template directions. 

\begin{algorithm}[htp]
    \caption{Reachability analysis over a compositional path $\rho_{shallow}$.}
    \label{Algo:path-wise-RA.}
    \scriptsize
    \begin{algorithmic}[1]
    \Require $\rho_{shallow} = \langle \rho_1,\rho_2,\dots,\rho_m \rangle$, CAHA $\H$, Initial and unsafe configuration $\I=(v^{init},\C^{init})$, $\F = (v^{unsafe}$, $~\C^{unsafe})$.
    \Ensure Unknown when feasible, otherwise returns safe with the infeasible subpath $\rho'_{inf}$.
    \State $R_{\rho_{shallow}} = \emptyset$\Comment{$R_{\rho_{shallow}}$ stores the reachable states of $\rho_{shallow}$.}
    \State $\rho'_1=\rho'_2=\dots\rho'_m = \{\} ; \textbf{Q}= \{\}$\Comment{$\{\}$ denotes empty and \textbf{Q} is the BFS queue.}
    \For{$i = 1$ to $m$}
        \State $\rho'_i.push(v_i^{init})$\Comment{$v^{init}=\langle v_1^{init},v_2^{init},\dots,v_m^{init}\rangle$}
    \EndFor
    \State $\rho'_{inf} = \langle \rho'_1,\rho'_2,\dots,\rho'_m \rangle$\Comment{Monitoring infeasible path segments.}
    \State $\textbf{loc} = \langle v_1^{init}, v_2^{init}, \ldots, v_m^{init}\rangle$ ; $ \textbf{C} = \C^{init}_1 \times \C^{init}_2 \times \ldots \times \C^{init}_m$ \label{initial}\Comment{Create initial symbolic state}
    \State $\textbf{Q}.push\big((\textbf{loc}, \textbf{C})\big)$ \label{queue_define}\Comment{Initially queue stores the initial symbolic state.}
    \While{($\textbf{Q}\neq \emptyset$)}
        \label{while_start}
        \State $(\textbf{loc}, \textbf{C}) = \textbf{Q}.remove()$ 
        \If{($flowpipe\_map[(\textbf{loc}, \textbf{C})]== Null$)}\label{f_cache}\Comment{Checking if the $flowpipe\_map$ has no entry for (\textbf{loc}, \textbf{C}).}
            \State $(\textbf{loc},\textbf{C}^{flow}) = postC\big((\textbf{loc},\textbf{C})\big)$ ; $R_{\rho_{shallow}} = R_{\rho_{shallow}} \cup (loc,\textbf{C}^{flow})$ \label{algo:flowpipe}\Comment{Compute reachable states.} 
            \State $flowpipe\_map[(\textbf{loc}, \textbf{C})]= (\textbf{loc},\textbf{C}^{flow})$\label{flow_store}\Comment{Store reachable states for (\textbf{loc}, \textbf{C}).}
        \Else
            \State $(\textbf{loc},\textbf{C}^{flow}) = flowpipe\_map[(\textbf{loc},\textbf{C})]$\label{entry_flowpipe}\Comment{Retrieved reachable states from the $flowpipe\_map$.}
        \EndIf
        \State $\Sigma_{next} = \emptyset$; $no\_next = true$
        \For{$i = 1$ to $m$}
            \label{algo:next_trans} 
            \If{$\textbf{loc}$[i] has next transition in $\rho_i$ with label $\omega_i$}
               \State $\Sigma_{next} = \Sigma_{next}\cup \omega_i$ ,  $no\_next = false$ \Comment{Find the set of possible next transitions.}
            \EndIf
        \EndFor\label{algo:next_trans_end}
        \If{$no\_next$}\label{algo:reachable}
            \If{$R \cap \C^{unsafe}$}\Comment{Check safety violation.}
               \State \Return unknown with feasible path $\rho_{shallow}$.            
            \Else
                \State \Return safe with $\rho_{shallow}$
            \EndIf
        \EndIf\label{reachable_end}
        \For{$\omega$ in $\Sigma_{next}$}\Comment{Process each transition.} \label{postD_begins}
            \State $\textbf{next\_loc} = \textbf{loc}$ ; $trans = \emptyset$ \Comment{$trans$ is the set of transitions that can be executed simultaneously.}
            \For{$i = 1$ to $m$}\label{next_loc}\Comment{Identify participating components and construct next location.}
                \If{($\omega \notin \Sigma_{sh}$) or ($\omega \in \Sigma_{sh} \land Frq(\omega) == Sh\_comps(\omega)$) and $(\textbf{loc}[i],\omega,v)$ is a transition in $\rho_i$}
                   \State $\textbf{next\_loc}$[i] = $v$ ; $trans \gets trans \cup \delta$\Comment{$\delta = (\omega,\textbf{loc}[i],\G,\Asgn,v)\in Trans$}
                    \State $\rho'_i.push(\omega)$ ;  $\rho'_i.push(v)$
                \EndIf
            \EndFor \label{next_loc_end}
            \If{$trans == \emptyset$}
                \State \textbf{continue} \label{not_compatible}\Comment{Transition with label $\omega$ is shared and cannot be enabled.}
            \Else
                \State $\G, \Asgn = make\_compatible(trans)$\label{compatible}\Comment{Transition(s) with label $\omega$ is enabled and either local or shared.}
            \EndIf
            \If{$successor\_map[\big((\textbf{loc},\textbf{C}^{flow}), \omega\big)]== Null$}\label{successor_check}\Comment{ Checking if $successor\_map$ has no entry for $\big((\textbf{loc},\textbf{C}^{flow}), \omega\big)$.} 
                \State ($\textbf{next\_loc},\textbf{C}^{new}$) = $postD\big((\textbf{loc},\textbf{C}^{flow}),\G,\Asgn,\textbf{next\_loc}\big)$\label{next_SS}
                \If{$\textbf{C}^{new} \neq \emptyset$}
                    \State \textbf{Q}.push\big(($\textbf{next\_loc},\textbf{C}^{new}$)\big) \Comment{Insert the newly computed symbolic state into the queue.}
                    \State $successor\_map[\big((\textbf{loc},\textbf{C}^{flow}), \omega\big)] = (\textbf{next\_loc},\textbf{C}^{new})$\Comment{Store the symbolic state in $successor\_map$.}
                \EndIf
            \Else
                \State $(\textbf{next\_loc},\textbf{C}^{new})= successor\_map[\big((\textbf{loc},\textbf{C}^{flow}), \omega\big)]$\label{reuse_successor}\Comment{Retrieved symbolic state from the $successor$ map.}
                \State \textbf{Q}.push\big(($\textbf{next\_loc},\textbf{C}^{new}$)\big)
            \EndIf
        \EndFor\label{postD_end}
    \EndWhile\label{while_end}
    \State \Return safe with $\rho'_{inf}$
    \end{algorithmic}
\end{algorithm}

\paragraph{Computation Tree Exploration}
\noindent Consider the shallow compositional path $\rho_{shallow}$ depicted in Fig. \ref{fig:shallowpath}. We now discuss the procedure for computing $R_{\rho_{shallow}}$. The initial locations of the three components are merged to form the location ($\langle1,5,8\rangle$) in the composed automaton, where the set of variables $\mathcal{X}$ and the $Inv$, $Init$, and $Flow$ maps of each component are combined as defined in Definition \ref{def:composed_ha}. We compute the reachable states for this location using the $postC$ operator, detailed in \cite{10.1145/3567425}. The computed reachable states are represented by a symbolic state ($\langle1,5,8\rangle$, \textbf{C}) and stored in a data structure as the root of the computation tree. Subsequently, the algorithm examines the possible transitions that can be activated to advance to the next location in the composed automaton. The transition dynamics is obtained by merging the transition dynamics of the components according to the rules specified in Definition \ref{def:composed_ha}. 
From the initial location $\langle1,5,8\rangle$, there are three outgoing transitions labeled $l_1$, $sh_1$, and $l_3$. The transition with label $sh_1$ is not considered, as it is a shared transition and is not synchronized with the first component. The algorithm therefore identifies two transitions with labels $l_1$ and $l_3$ and constructs the two target locations $\langle2,5,8\rangle$ and $\langle1,5,9\rangle$, respectively.  These transitions create two branches in the computation tree, leading to the computation of two symbolic states ($\langle2,5,8\rangle$, $\textbf{C}_1$) and ($\langle1,5,9\rangle$, $\textbf{C}_2$) using the $postD$ operator, representing the discrete transition successors. The computation tree expands as the algorithm explores reachable states from the newly computed symbolic states. For example, from location $\langle2,5,8\rangle$, two transitions with labels $sh_1$ and $l_3$ construct two new symbolic states ($\langle3,6,8\rangle$,$\textbf{C}_3$) and ($\langle2,5,9\rangle$,$\textbf{C}_4$) using $postC$ followed by $postD$ operations. The computation tree of the shallow compositional path is illustrated in Fig. \ref{fig:computation_tree}. This iterative process computes the reachable state-space for a shallow compositional path.  If the reachable states intersect with the unsafe region, our algorithm terminates by showing that the safety specification is violated. Otherwise, the algorithm returns an infeasible subpath and continues processing the next path for further analysis. 

\begin{algorithm}[htp]
    \caption{$postD((\textbf{loc},\textbf{C}^{flow}),\G,\Asgn,\textbf{next\_loc})$}
    \label{postD_algo2}
\begin{algorithmic}[1]
    \footnotesize
    \Require A symbolic state $(\textbf{loc},\textbf{C}^{flow})$, and $\G, \Asgn$ of a transition from $\textbf{loc}$ to $\textbf{next\_loc}$.
    \Ensure A new symbolic state ($\textbf{next\_loc},\textbf{C}^{new}$).
    \State $\C' = \emptyset$ ; $\C'' = \emptyset$
    \For{$\Omega$ in $\textbf{C}^{flow}$}\label{postDstart}
        \State $\C' \gets \C' \cup \big(\Omega \wedge \G\big)$\label{g_intersect}\Comment{Set satisfying the guard $\G$.}
    \EndFor\label{g_inter_end}
    \For{$\Omega$ in $\C'$}\label{intersected_region}
        \State $\Omega^* = \Omega.map(\Asgn)$\label{applyingreset}\Comment{Elements of $\Omega^*$ and $\Omega$ satisfy $\Asgn$.}
        \State $\C'' \gets \C'' \cup \big( \Omega^* \wedge \Inv(\textbf{next\_loc})\big)$\label{Inter_inv}
    \EndFor
    \State $\textbf{C}^{new}= \underset{\ell\in D} {\bigwedge}\ell\cdot x \leq \underset{\Omega\in\C''}{\text{max}}(h_{\Omega}(\ell))$ \label{postD}\Comment{Convert to template polyhedron.}
    \State \Return ($\textbf{next\_loc},\textbf{C}^{new}$)
\end{algorithmic}
\end{algorithm}

\begin{wrapfigure}[23]{r}{0.4\textwidth}
    \centering
    \includegraphics[scale=0.35]{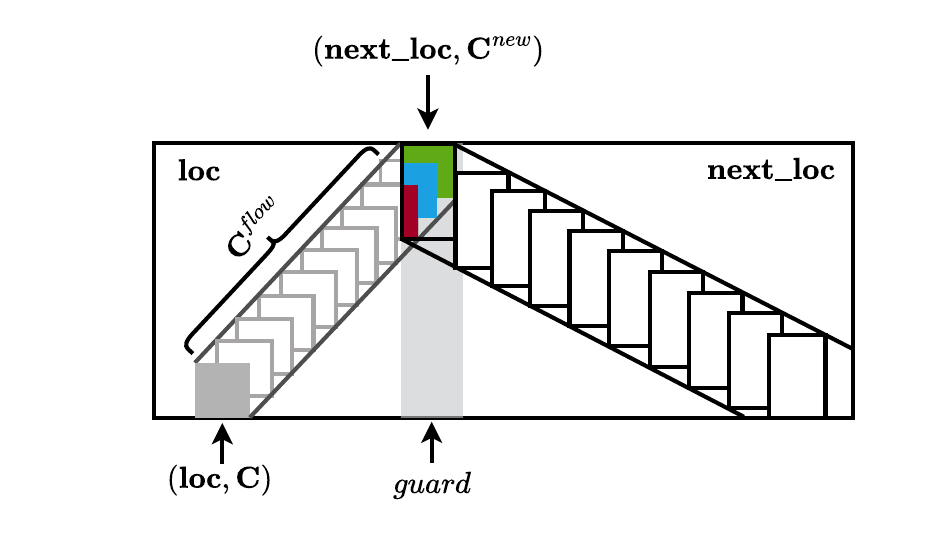}
    \caption{Illustration of the reachable states computation. ($\textbf{loc},\textbf{C}$) is the initial symbolic state, where $\textbf{C}$ is the initial set. The initial set $\textbf{C}$ evolves over time according to the $Flow$ of $\textbf{loc}$, resulting in new reachable states $\textbf{C}^{flow}$. The polyhedra in $\textbf{C}^{flow}$ are intersected with the guard $\G$, and the resulting regions (shown in different colors) satisfy both the reset assignment $\Asgn$ and the $\Inv$ of the location $\textbf{next\_loc}$. These intersected regions are then combined into a single over-approximation convex set $\C^{new}$ using template hull aggregation to form a new symbolic state $(\textbf{next\_loc}, \textbf{C}^{new})$.}
    \label{fig:postC_demo}
\end{wrapfigure}

\noindent The procedure of computing the reachable state space in $\rho_{shallow}$ is shown in Algorithm \ref{Algo:path-wise-RA.}. Line \ref{initial} builds the initial symbolic state $(\textbf{loc},\textbf{C})$ of the composed automata from the safety specification. Line \ref{queue_define} defines a queue $\textbf{Q}$ that contains the initial symbolic state. We explore the symbolic states in a breadth-first fashion with the help of this queue. The while loop in lines \ref{while_start} - \ref{while_end} finds the next reachable symbolic states and pushes them into the queue. The loop is executed until the queue is empty. The $flowpipe\_map$ is a data structure that contains the reachable symbolic state for $(\textbf{loc},\textbf{C})$. Line \ref{f_cache} checks the $flowpipe\_map$ to determine whether the reachable states for the symbolic state $(\textbf{loc},\textbf{C})$ are computed. If not, line \ref{algo:flowpipe} computes the set of reachable states of the current symbolic state using the $postC$ operation and stores it in $R_{\rho_{shallow}}$. Additionally, the reachable states are stored in the $flowpipe\_map$ for reuse, stated in line \ref{flow_store}. When the $flowpipe\_map$ has an entry for the key $(\textbf{loc},\textbf{C})$, the reachable states are retrieved from the map, defined in line \ref{entry_flowpipe}. A set of labels $\Sigma_{next}$ of the possible outgoing transitions is constructed from $\textbf{loc}$, depicted in lines \ref{algo:next_trans} to \ref{algo:next_trans_end}. In lines \ref{algo:reachable} to \ref{reachable_end}, the algorithm terminates by concluding $\H$ as unknown if the reachable states $R_{\rho_{shallow}}$ intersect with $\C^{unsafe}$, otherwise it returns safe and the path $\rho_{shallow}$ as infeasible. The lines \ref{postD_begins} - \ref{postD_end} find the next symbolic states $(\textbf{next\_loc},\textbf{C}^{new})$ for every $\omega$ in the set $\Sigma_{next}$ and push them to the queue $\textbf{Q}$. The next location in the composed automata for the transition(s) with label $\omega$ is constructed in lines \ref{next_loc} - \ref{next_loc_end}. We use the function $Sh\_comps(\omega)$ that finds the number of components where $\omega$ is shared and $Frq(\omega)$ that returns the frequency of the label $\omega$ in $\Sigma_{next}$. When the transition with label $\omega$ is identified as a shared transition, but $Frq(\omega)$ does not equal $Sh\_comps(\omega)$, it indicates that the transitions for $\omega$ are not fully synchronized and therefore cannot proceed. In such cases, the algorithm refrains from computing the next symbolic state for $\omega$ and proceeds to evaluate the subsequent transition, as defined in line \ref{not_compatible}. If the $\omega$ is a label of a local transition, then the size of the set $trans$ must be $1$. The function $make\_compatible()$ depicted in line \ref{compatible} ensures the guard $\G$ and reset assignment $\Asgn$ of the transitions in $trans$ are compatible with the composed automata as per Definition \ref{def:composed_ha}. The $successor\_map$ is a data structure that contains the successor symbolic state for the pair $\big((\textbf{loc},\textbf{C}^{flow}), \omega\big)$. Line \ref{successor_check} checks whether the new symbolic state is already computed for the $\omega$ from the current symbolic state ($\textbf{loc},\textbf{C}^{flow}$). If not, line \ref{next_SS} computes the new symbolic state ($\textbf{next\_loc},\textbf{C}^{new}$) using $postD$ operation, described in Algorithm \ref{postD_algo2}. If $\textbf{C}^{new}$ is empty, it indicates that the algorithm has encountered an infeasible segment of the path. In such cases, the algorithm proceeds to consider other discrete transitions. Otherwise, it stores the new symbolic state ($\textbf{next\_loc},\textbf{C}^{new}$) in the queue $\textbf{Q}$ and the $successor\_map$. If the $successor\_map$ has an entry for the key $\big((\textbf{loc},\textbf{C}^{flow}), \omega\big)$, the new symbolic state ($\textbf{next\_loc},\textbf{C}^{new}$) is already computed. Line \ref{reuse_successor} retrieves the symbolic state from the $successor\_map$.

\begin{wrapfigure}{r}{0.3\textwidth}
     \centering
    \includegraphics[scale=0.5]{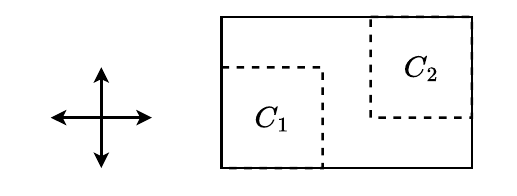}
    \caption{The outer rectangle region is the computed template hull aggregation from two convex sets $C_1$ and $C_2$, along the \emph{axis} directions as template directions.}
    \label{fig:tem_poly}
\end{wrapfigure}

\noindent Algorithm \ref{postD_algo2} computes the next symbolic state for the transition(s) in $trans$. The lines \ref{postDstart} to \ref{postD} compute the successor states due to the discrete transition(s) with label $\omega$. The for loop in lines \ref{postDstart}- \ref{g_inter_end} finds the bounded convex sets $\C'$ after checking the intersection of the template polyhedra $\textbf{C}^{flow}$ with the guard $\G$. Line \ref{intersected_region} iterates over the obtained convex sets that satisfy the guard $\G$. Line \ref{applyingreset} constructs a convex set $\Omega^*$ which is obtained by applying the reset assignment $\Asgn$ to the resultant set, while line \ref{Inter_inv} identifies those sets that satisfy the $Inv$ of the target location $\textbf{next\_loc}$ and are stored in $\C''$. The emptiness of $\C'$ and $\C''$ indicates that the guard $\G$ and reset assignment $\Asgn$ of the transition(s) are not satisfied, respectively. There may be many polyhedra in $\C''$, and computing the $postC$ successors in the target location for each of them can be expensive. Therefore, set-aggregation methods are generally applied. We use the \emph{template hull aggregation} technique. As an illustration, the template hull of two sets $C_1$ and $C_2$ is shown in Fig. \ref{fig:tem_poly}. The line \ref{postD} computes the reachable states $\mathcal{C}^{new}$ for location $\textbf{next\_loc}$ by applying \emph{template hull aggregation} to the sets in $\C''$. Figure \ref{fig:postC_demo} shows an overview of the $postC$ and $postD$ computation.

\subsection{Eliminating Infeasible Paths from Enumeration}\label{prune_section}

\noindent Once we have performed reachability analysis on $\rho_{shallow}$ and find that it is infeasible, our algorithm returns the infeasible subpath $\rho'_{inf}$ of $\rho_{shallow}$. We construct a negation constraint corresponding to the stuttered version (say $\rho'_{step}$) of the sub-path $\rho'_{inf}$ and conjunct it with the bounded graph encoding. Considering the length of the subpath $\rho'_{step}$ is $k$, the path negation constraint is: 
$$\phi_{neg}(\rho_{step'}) = \neg(\bigwedge_{1 \leq c \leq m}\bigwedge_{t_c \in \rho_c, 1 \leq j \leq k}{t^j_{c}})$$

\noindent For every infeasible sub-path identified by reachability analysis, the above negation constraint is conjoined with $\phi_{\H}^k$. The infeasible sub-path always starts from the initial locations.\\
\textbf{Remark:} The path negation constraint additionally discards all paths which extends $\rho'_{inf}$. This is sound and serves as an optimization because extensions of an infeasible path are also infeasible.

\subsection{Optimizations in Path Enumeration}
Additionally, the following two optimizations proposed in \cite{10.1007/978-3-030-94583-1_23} are integrated in the encoding of $\phi^k_{\H}$, namely global waiting and random waiting:

\noindent \textbf{Global waiting.} \emph{global waiting} refers to the situation when all component paths wait simultaneously with a stutter transition. This introduces redundant paths, which we eliminate by introducing the $\phi_{G\_waiting}$ constraint.
$$\phi_{G\_waiting}= \bigwedge_{1\leq j \leq k} \bigvee_{1\leq c \leq m} \big(\neg stutter_c^j\big)$$

\noindent \textbf{Random waiting.} Stutter transitions are intended to keep a component waiting for synchronization with shared transitions and therefore should occur only before a shared transition in a path. However, the encoding allows stutters to occur before location transitions. We therefore include a constraint $\phi_{R\_waiting}$ to eliminate paths where stutter transitions occur before local transitions.
\begin{equation}
    \phi_{R\_waiting}= \bigwedge_{1\leq c \leq m} \bigwedge_{1 \leq j \leq k-1} \Biggl( stutter_c^j \implies \\
    \bigvee_{\{t_c \in E_c |label(t_c) \in (\Sigma - \Sigma_{sh}) \}} \big(\neg t_c^{j+1}\big)\Biggr)
\end{equation}

\subsection{Shared Reachability Analysis}
\begin{wrapfigure}{r}{0.4\textwidth}
    \centering
    \vspace{-25pt}
    \includegraphics[width=1\linewidth]{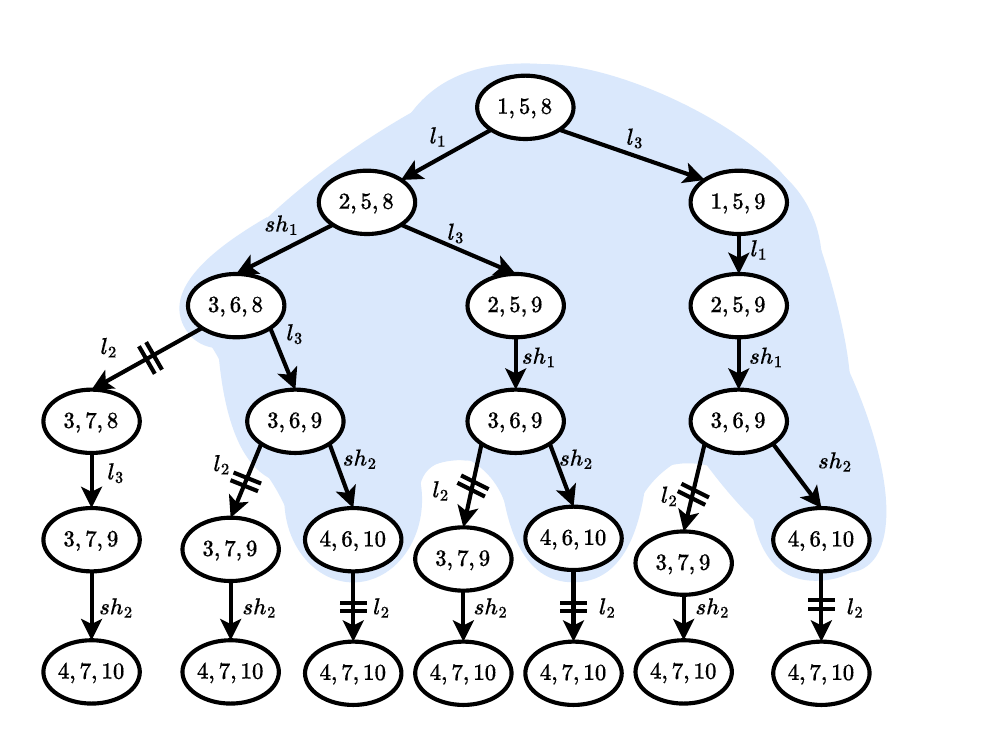}
    \caption{The reachable states of the shaded sub-tree is memoized in data structures during the reachability analysis of $\rho_{shallow}$. This can be subsequently reused during the reachability analysis of  $\rho_{shallow'}$.}
    \vspace{-10pt}
    \label{fig:computation_tree}
\end{wrapfigure}

The computation tree of shallow paths may often have common sub-trees. Computing one can therefore expedite the computation of the other. We propose a caching technique that memoizes reachable states in data structures, enabling their reuse for subsequent paths, thus reducing redundant computations. As a motivational example, consider the safety specification discussed in Section \ref{motivating_exam}. Our algorithm identifies two shallow compositional paths $\rho_{shallow}$ and $\rho_{shallow'}$ depicted in figure \ref{fig:shallowpath} and \ref{fig:shallowpath2}, respectively, for the bound of analysis $5$. Assume that one of these shallow paths, say $\rho_{shallow}$, is considered for reachability analysis, and found to be infeasible because the \emph{guard} of the transition with label $l_2$ is not satisfied. The algorithm then proceeds to the next path $\rho_{shallow'}$ for its reachability analysis. 

\begin{wrapfigure}[9]{r}{0.35\textwidth}
    \centering
    \vspace{-18pt}
    \includegraphics[width=1\linewidth]{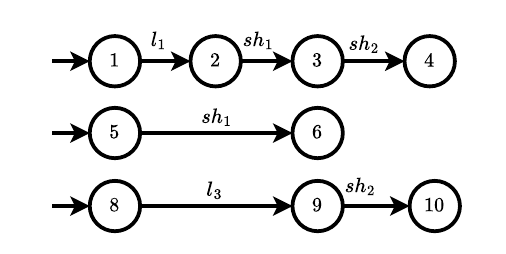}
    \caption{A shallow compositional subpath $\rho_{sub}$ of paths $\rho_{shallow}$ and $\rho_{shallow'}$. }
    \label{fig:subpath}
\end{wrapfigure}

\noindent During the analysis of $\rho_{shallow'}$, the reachable states computed for $\rho_{shallow}$, particularly those corresponding to the common sub-path depicted in Fig. \ref{fig:subpath}, can be reused. Figure \ref{fig:computation_tree} shows the computation tree for the paths $\rho_{shallow}$ and $\rho_{shallow'}$. The shaded part highlights the common sub-tree that can be reused during the reachability analysis of the path $\rho_{shallow'}$.

\noindent To efficiently reuse the reachable states and avoid redundancy in computation, we used two map data structures, namely \emph{flowpipe} and \emph{successor}. 
\begin{itemize}
    \item The \emph{flowpipe\_map} stores the reachable states $(\textbf{loc},\textbf{C}^{flow}) = postC(\textbf{loc},\textbf{C})$ from a symbolic state (\textbf{loc},\textbf{C}).
    \item The \emph{successor\_map} stores the successor symbolic state ($\textbf{next\_loc},\textbf{C}^{new}$) for a pair \big($(\textbf{loc},\textbf{C}^{flow}),$ $\omega$\big), where $\omega$ is a label of a transition between $\langle\textbf{loc},\textbf{next\_loc}\rangle$ in the composed automata.
\end{itemize}

\begin{claim}
    Given a compositional hybrid automaton $\mathcal{H}$, a safety specification, and a bound of analysis $k$, the proposed approach correctly determines the safety of the system.
\end{claim}
\begin{proof}
    The path enumeration procedure depicted in algorithm \ref{algo:path_enum} ensures complete coverage of all step compositional paths, which collectively capture all possible interleaving compositional paths of length bounded by $k$. For each enumerated step compositional path, the proposed algorithm constructs a corresponding shallow compositional path and performs symbolic reachability analysis along this path. The reachability computation relies on the $postC$ and $postD$ operators, as defined in~\cite{FLGDCRLRGDM11}. Since these operators are proven to be sound over-approximations for computing reachable states, the correctness of our algorithm directly inherits from these results.
\end{proof}

%% file: results.tex
\section{Results}
\subsection{\textbf{Experimental Setup}}
The experiments are performed on a 12-core Intel i7-12700K, $4.90$ GHz CPU with 64 GB RAM on Ubuntu $22.04$ OS. The proposed algorithm is implemented in the tool \textsc{SAT-Reach}, and we refer to this as \textsc{SAT-Reach-C}, the C standing for support of compositional models. We evaluate the performance of our proposed algorithm on several compositional affine and linear hybrid automata. A comparative performance evaluation is conducted with a state-of-the-art model checking tool \textsc{SpaceEx}, a symbolic reachability analysis tool for linear and piecewise affine hybrid systems, \textsc{XSpeed} and \textsc{SAT-Reach}, both are flowpipe-based reachability analysis tools for affine hybrid systems, and \textsc{BACH}, an SMT-based BMC tool for linear hybrid systems. In addition, we present a comparison with a well-known bounded reachability analysis tool for non-linear hybrid systems \textsc{dReach}. The experiments have a $20$ minute timeout and a 16 GB memory limit. Timeout and OOM indicate the instance could not produce results within the time and memory limits, respectively. The input format of \textsc{SAT-Reach-C} is the same as that of \textsc{SpaceEx}, with one model file describing the composition of hybrid automata and one configuration file describing the initial states, forbidden states, and parameter settings. The tables described below consist of \#Vars, \#Locations, and \#Trans, which represent the number of variables, locations, and transitions, respectively, in the product automaton of the CAHA or CLHA. Additionally, \#Comps denotes the number of components in the composed automata as shown in the following tables. The Bound shows the discrete jump of analysis in interleaving semantics for the respective instances. In our experiments, a system is deemed to be unsafe when the set of computed reachable states intersects with the unsafe states. The source code and repeatability instructions are available online in [\url{https://gitlab.com/Atanukundu/SAT-Reach/-/tree/dev}]. 

\subsection{Parameter Settings}
Performing a fair comparison among various tools requires carefully defining their parameter settings. The details of the parameter configurations are illustrated below.

\noindent \subsubsection{\textsc{dReach}}For \textsc{dReach}, we translate \textsc{SpaceEx}'s input language to \textsc{dReach} compatible input file (.drh) using the Hyst~\cite{bak2015hscc} model translation tool. For some models, the translation time overshoots the timeout threshold. We mark Timeout* for those instances in the tables. The rest of the instances take a reasonable time to convert the model files. In \textsc{dReach}, we have used the default $\delta = 0.001$.

\noindent \subsubsection{\textsc{SpaceEx}} The input language of \textsc{SAT-Reach-C} is the same as that of \textsc{SpaceEx}. We now briefly discuss the parameter settings, along with the model modifications that have been made:
\begin{itemize}
    \item LGG support function scenario is used for flowpipe computation in \textsc{SpaceEx}.
    \item We use the same \emph{sampling-time} and \emph{local-time-horizon} in \textsc{SAT-Reach-C} and \textsc{SpaceEx}.
    \item In both the tools \textsc{SAT-Reach-C} and \textsc{SpaceEx}, we used 100$\%$ \emph{clustering} with \emph{template hull aggregation}.
    \item We find the minimum value for the parameter \emph{Max. iterations} using trial-and-error, which ensures that forbidden states are reachable within that \emph{Max. iterations} for unsafe instances.
\end{itemize} 
\noindent To ensure that \textsc{SpaceEx} adheres to the BMC bound $k$ during state-space exploration, we introduce two new variables with constant dynamics in the model: $t\_count$ and $bound$. Initially \emph{bound} is set to $k$ and \emph{$t\_count$} is set to $0$. For every discrete transition, \emph{$t\_count$} is increased by $1$, and $t\_count <= bound - 1$ is checked in the \emph{guard}. Therefore, \textsc{SpaceEx} limits its state-space exploration to the BMC bound $k$.

\subsubsection{\textsc{XSpeed}, \textsc{SAT-Reach} and \textsc{SAT-Reach-C}}
We have used the same parameter settings for \textsc{SAT-Reach} and \textsc{XSpeed} as we used in \textsc{SAT-Reach-C}. We use a local time-horizon of 20 and a time-step of 0.01 for all instances except FDDI. For FDDI, we use a local time-horizon of 600 and 100 as the time-step.

\noindent \subsubsection{\textsc{BACHopt}}
\textsc{BACH} utilizes the same input language as \textsc{SpaceEx}'s input language. It leverages \textsc{cmsat} as a SAT solver for the mixed-semantics guided BMC procedure. \textsc{BACHopt} extends \textsc{BACH} by incorporating two additional optimization techniques, enhancing its scalability and efficiency.
\begin{itemize}
    \item \emph{nid}: Non-identical path guided path pruning
    \item \emph{rmults}: multiple infeasible path segments based pruning
\end{itemize}
The \textsc{BACHopt} tool is designed for linear hybrid automata models. This is why we show the result of \textsc{BACHopt} only for CLHA instances.

\subsection{Performance Evaluation on CAHA Benchmarks}
 
\noindent The performance of our tool \textsc{SAT-Reach-C} is evaluated using various CAHA benchmarks to assess its scalability and efficiency. The benchmark-wise detailed analysis is discussed below.

\subsubsection{Filtered Oscillator}

\begin{wraptable}{r}{0.35\textwidth}
        \centering
        \caption{Computation time of unsafe instances of filtered oscillator benchmark.}
        \resizebox{5cm}{!}{
        \begin{tabular}[c]{lccc}
                \hline
                {Instance} & FOU004 & FOU016 & FOU064 \\
                Safety & unsafe & unsafe & unsafe \\
                \# Vars & 6 & 18 & 66 \\
                \# Locations & 4 & 4 & 4 \\
                \# Comps & 5 & 17 & 65 \\
                \# Trans & 4 & 4 & 4 \\
                Bound & 3 & 3 & 3 \\
                \hline
                Tools & \multicolumn{3}{c}{Computation time in second} \\
                \hline
                \textsc{SpaceEx} & 0.05 & 0.24 & 3.01 \\
                \hline
                \textsc{dReach} & 1.14 & 12.18 & Error \\
                \hline
                \textsc{SAT-Reach} & 0.17 & 1.11 & 22.95 \\
                \hline
                \textsc{XSpeed} & 0.15 & 1.10 & 23.09 \\
                \hline
                \textsc{SAT-Reach-C} & 0.18 & 1.16 & 22.72 \\
                \hline
            \end{tabular}
        }
        \vspace{-10pt}
        \label{tab:FO}
    \end{wraptable}
    The Filtered Oscillator ~\cite{FLGDCRLRGDM11} models an oscillator with additional filters to reduce the oscillation amplitude. Filters increase the system dimension, testing the scalability of verification algorithms with the system dimension. The safety specification aims to avoid a designated region in the state-space. Table~\ref{tab:FO} represents the benchmark details and computational time required by each tool. The instance FOUn denotes the system consisting of $n$ filters, where U denotes that the instance is classified as an unsafe instance. Based on this table, we can conclude that \textsc{SpaceEx} demonstrates better scalability than the other tools. The composed automaton of the FOU64 instance contains only $4$ locations and $4$ transitions despite having $65$ components. In such cases, CEGAR-based state-space exploration does not show benefits when the composed automaton has a limited number of locations and transitions. 
   
\subsubsection{Navigation}
The NAV~\cite{FehnkerI04, DBLP:journals/sttt/BogomolovDFGJLP16} models the composition of hybrid automata of navigating objects in varying grid sizes. Affine ODEs, which may vary in each grid cell, describe the object's motion. When an object moves to a neighboring cell, the dynamics are updated instantaneously. The hybrid dynamics is expressed as a hybrid automaton, with a location for each grid cell and discrete transitions between adjacent locations. The different grid size leads to different benchmarks, such as NAVn, where n represents the grid size. Shared transitions have been introduced in the models proposed. The compositions have a large number of locations and transitions, testing the scalability of verification algorithms with the size of the product automaton. The safety specification prohibits entering certain cells in the grid from a designated initial region within a given cell. The tables \ref{tab:NAV3_unsafe} and \ref{tab:NAV3_safe} show the instance details along with the computational time required by the tools for the NAV3~\cite{DBLP:journals/sttt/BogomolovDFGJLP16} benchmark. In Tables \ref{tab:NAV3_unsafe} and \ref{tab:NAV3_safe}, the instances are generated by increasing the number of components while preserving the respective safety specifications, which leads to a complex system with a large number of locations and transitions. In table \ref{tab:NAV3_unsafe}, NAV3C2 represents a system that consists of two navigation systems of grid size $3$x$3$. As the number of locations and transitions in the composed automaton increases, most of the tools fail to classify the instances. \textsc{SAT-Reach} uses a path-guided state-space exploration algorithm on the composed automata and fails to classify the most complex instance (NAV3C5). \textsc{SpaceEx} is unable to classify the instances from NAV3C3 onwards, despite its ability to perform on-the-fly state-space exploration. In contrast, \textsc{SAT-Reach-C} efficiently classifies all such instances because of the CEGAR-based state-space exploration, without an explicit product automaton construction, in which the abstract counterexample targets only the relevant discrete state-space rather than exploring the entire state-space exhaustively. This approach makes \textsc{SAT-Reach-C}efficient and scalable among these tools.

\begin{table}[htp]
    \centering
    \begin{minipage}{0.49\linewidth}
        \centering
        \caption{Time required to classify unsafe instances of Navigation (NAV3) benchmarks.}
        \resizebox{0.8\linewidth}{!}{
    \begin{tabular}[c]{lcccc}
        \hline
        {Instance} & NAV3C2 & NAV3C3 & NAV3C4 & NAV3C5 \\
        Safety & unsafe & unsafe & unsafe & unsafe\\
        \# Vars & 8 & 12 & 16 & 20 \\
        \# Locations & 81 & 729 & 6561 & 59049 \\
        \# Comps & 2 & 3 & 4 & 5 \\
        \# Trans & 378 & 5103 & 61236 & 688905 \\
         Bound & 6 & 9 & 12 & 15 \\
        \hline
        Tools & \multicolumn{4}{c}{Computation time in second} \\
        \hline
        \textsc{SpaceEx} & 73.52 & Timeout & Timeout & Timeout\\
        \hline
        \textsc{dReach} & 22.88 & OOM & OOM & OOM \\
        \hline
        \textsc{SAT-Reach} & 1.07 & 4.09 & 310.92 & Timeout \\
        \hline
        \textsc{XSpeed} & 164.79 & Timeout & Timeout & Timeout\\
        \hline
        \textsc{SAT-Reach-C} & 0.50 & 1.33 & 2.55 & 3.37 \\
        \hline
    \end{tabular}
    }
    
    \label{tab:NAV3_unsafe}
    \end{minipage}
    \hfill
    \begin{minipage}{0.49\linewidth}
        \centering
        \caption{Time required to classify safe instances of Navigation (NAV3) benchmarks.}
        \resizebox{0.8\linewidth}{!}{
    \begin{tabular}[c]{lcccc}
        \hline
        {Instance} & NAV3C2 & NAV3C3 & NAV3C4 & NAV3C5 \\
        Safety & safe & safe & safe & safe\\
        \# Vars & 8 & 12 & 16 & 20 \\
        \# Locations & 81 & 729 & 6561 & 59049 \\
        \# Comps & 2 & 3 & 4 & 5 \\
        \# Trans & 378 & 5103 & 61236 & 688905 \\
         Bound & 4 & 6 & 8 & 10 \\
        \hline
        Tools & \multicolumn{4}{c}{Computation time in second} \\
        \hline
        \textsc{SpaceEx} & 9.88 & Timeout & Timeout & Timeout\\
        \hline
        \textsc{dReach} & 1.36 & 67.27 & OOM & OOM \\
        \hline
        \textsc{SAT-Reach} & 0.22 & 1.58 & 95.28 & Timeout \\
        \hline
        \textsc{XSpeed} & 45.59 & Timeout & Timeout & Timeout \\
        \hline
        \textsc{SAT-Reach-C} & 0.22 & 0.65 & 3.39 & 60.7 \\
        \hline
    \end{tabular}
    }
    
    \label{tab:NAV3_safe}
    \end{minipage}
\end{table}

\noindent Table \ref{tab:NAV_25} and \ref{tab:NAV_20} show the empirical results on other Navigation benchmarks, such as NAV25 and NAV20. The instances with U and S for both benchmarks indicate that these instances are known to be unsafe and safe instances, respectively. These two navigation systems are taken from \cite{DBLP:journals/sttt/BogomolovDFGJLP16} and \cite{FehnkerI04}, respectively. For these instances, \textsc{SAT-Reach}, and \textsc{XSpeed} fail to classify all instances except NAV20S1, whereas \textsc{SpaceEx} classifies 3 out of 8 instances within the time limit. In contrast, \textsc{SAT-Reach-C} successfully classifies all instances. As shown in Table~\ref{tab:NAV_25} and \ref{tab:NAV_20}, \textsc{SAT-Reach-C} demonstrates better scalability, efficiently handling large state-space within the time limit.

\begin{table}[htp]
    \centering
    \begin{minipage}{0.48\linewidth}
        \centering
        \caption{Scalability evaluation of NAV25 benchmarks.}
        \resizebox{0.8\linewidth}{!}{
    \begin{tabular}[c]{lcccc}
        \hline 
        {Instance} & NAV25U1 & NAV25U2 &  NAV25S1 &  NAV25S2\\
        Safety & unsafe & unsafe & safe & safe\\
        \# Vars & 8 & 8 & 8 & 8 \\
        \# Locations & 390625 & 390625 & 390625 & 390625 \\
        \# Comps & 2 & 2 & 2 & 2 \\
        \# Trans & 390600 & 390600 & 390600 & 390600  \\
        Bound & 4 & 25 & 4 & 8 \\
        \hline
        Tools & \multicolumn{4}{c}{Computation time in second} \\
        \hline
        \textsc{SpaceEx} & Timeout & Timeout & 304.12 &  Timeout\\
        \hline
        \textsc{dReach} & Timeout & OOM & Timeout & OOM \\
        \hline
        \textsc{SAT-Reach} & Timeout &  Timeout & Timeout & Timeout\\
        \hline
        \textsc{XSpeed} & Timeout & Timeout & Timeout & Timeout \\
        \hline
        \textsc{SAT-Reach-C} & 2.57 &  661.7 & 4.61 & 298.77 \\
        \hline
    \end{tabular}
    }
    
    \label{tab:NAV_25}
    \end{minipage}
    \hfill
    \begin{minipage}{0.48\linewidth}
        \centering
        \caption{Scalability evaluation of NAV20 benchmarks.}
        \resizebox{\linewidth}{!}{
    \begin{tabular}[c]{lcccccc}
        \hline
        {Instance} & NAV20U1 & NAV20U2 & NAV20U3 & NAV20S1 & NAV20S2\\
        Safety & unsafe & unsafe & unsafe & safe & safe \\
        \# Vars & 11 & 11 & 11 & 11 & 11\\
        \# Locations & 9729 & 9729 & 9729 & 9729 & 9729\\
        \# Comps & 2 & 2 & 2 & 2 &  2 \\
        \# Trans & 19204 & 19204 & 19204 & 19204 & 19204\\
        Bound & 30 & 47 & 38 & 9 & 15 \\
        \hline
        Tools & \multicolumn{5}{c}{Computation time in second} \\
        \hline
        \textsc{SpaceEx} & Timeout & Timeout & Timeout & 31.56 & 157.5\\
        \hline
        \textsc{dReach} & Timeout* & Timeout* & Timeout* & Timeout* & Timeout* \\
        \hline
        \textsc{SAT-Reach} & Timeout & Timeout & Timeout & 84.08 & Timeout\\
        \hline
        \textsc{XSpeed} & Timeout & Timeout & Timeout & 46.88 & Timeout \\
        \hline
        \textsc{SAT-Reach-C} & 14.91 & 89.93 & 13.94 & 2.39 & 12.37 \\
        \hline
    \end{tabular}
    }
    
    \label{tab:NAV_20}
    \end{minipage}
\end{table}

\begin{figure}
    \centering
    \begin{subfigure}[b]{.31\textwidth}
      \centering
	  {\includegraphics[width=\textwidth]{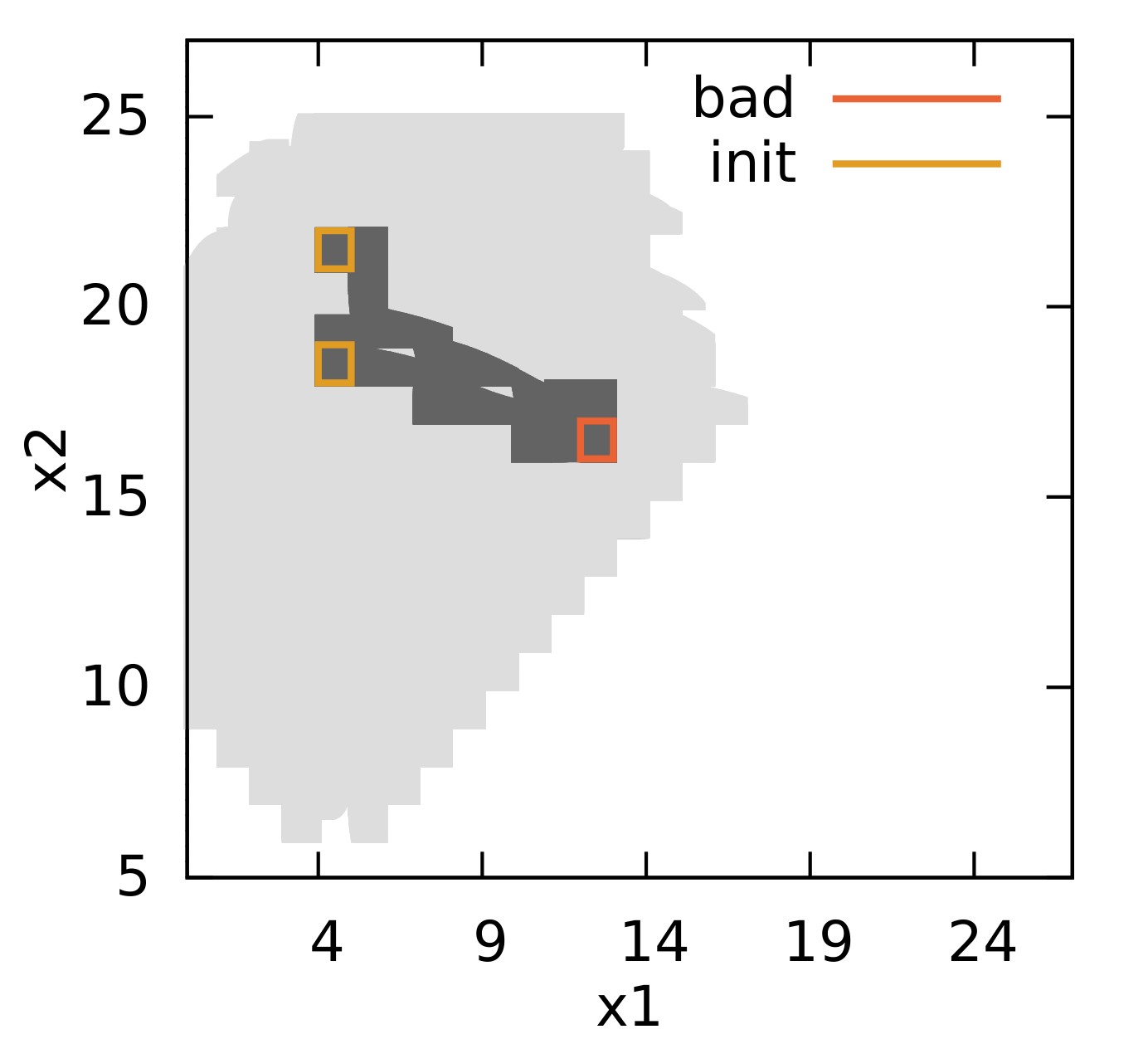}}
    \caption{ Explored reachable state-space for instance NAV\_25\_U2.}
    \label{fig:NAV_25_U3}
    \end{subfigure}%
    \hspace{3pt}
    \begin{subfigure}[b]{.31\textwidth}
      \centering
	  {\includegraphics[width=\textwidth]{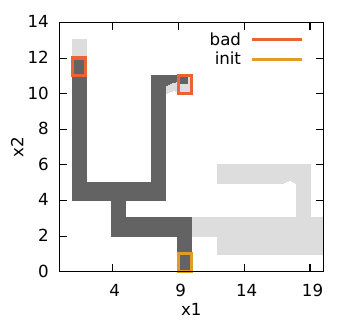}}
      \caption{Explored reachable state-space for an instance NAV\_20\_U1.}
    \label{fig:NAV_20_U1}
    \end{subfigure}%
     \hspace{3pt}
    \begin{subfigure}[b]{.31\linewidth}
      \centering
	  {\includegraphics[width=\textwidth]{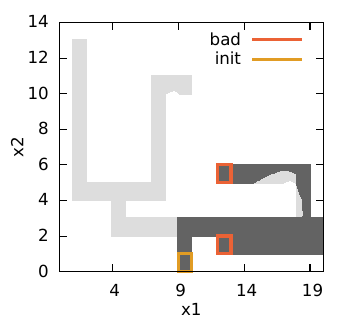}}
      \caption{Explored reachable state-space for an instance NAV\_20\_U3.}
    \label{fig:NAV_20_U1}
    \end{subfigure}%
    \caption{A comparison of computed reachable state-space among the tools \textsc{SpaceEx}, \textsc{XSpeed}, and \textsc{SAT-Reach-C}.}
    \label{fig:reachable_set}
\end{figure}

\noindent Figure \ref{fig:reachable_set} shows a comparison of explored reachable states by \textsc{XSpeed}, \textsc{SpaceEx}, and \textsc{SAT-Reach-C}. The distinct yellow and red boxes show the initial and unsafe regions of the two components, respectively. In figure \ref{fig:NAV_25_U3}, \textsc{SAT-Reach-C} shows a directed exploration (dark gray) leading towards the unsafe region. In contrast, \textsc{SpaceEx} computes the reachable states (light gray) in an uninformed way with breadth-first exploration. Similarly, in figure \ref{fig:NAV_20_U1}, \textsc{SAT-Reach-C} computes states (dark gray) towards the unsafe region, while \textsc{XSpeed} explores a wider set of states with breadth-first search (light gray). These figures demonstrate the benefits of CEGAR-based state-space exploration, in which the abstract counterexample directs the focus along the relevant discrete state-space, followed by its reachability analysis.

\subsection{Performance Evaluation on CLHA benchmarks}
We present here a performance evaluation of our tool on CLHA benchmarks and a comparison with CLHA supporting tools. These benchmarks are taken from ARCH COMP 2024~\cite{ARCH-COMP24}, a friendly competition in the PCDB category.

\begin{table}[htb]
	\setlength{\tabcolsep}{4pt}
	\renewcommand{\arraystretch}{1.2}
	\centering
	\caption{Computation times of the Nuclear Reactor System Benchmark.}
    \begin{adjustbox}{max width=0.9\textwidth}
	\begin{tabular}[c]{lcccccccccccc}
	\toprule
    {instance} & NRSS05 & NRSU05 & NRSS06 & NRSU06 & NRSS07 & NRSU07 & NRSS08 & NRSU08 & NRSS09 & NRSU09 & NRSS10 & NRSU10  \\ \midrule 
    safety & safe & unsafe & safe & unsafe & safe & unsafe & safe & unsafe & safe & unsafe & safe & unsafe \\
             \# Vars & 6 & 6 & 7 & 7 & 8 & 8 & 9 & 9 & 10 & 10 & 11 & 11  \\
         \# Locations & 1458 & 1458 & 5103 & 5103 & 17469 & 17469 & 59049 & 59049 & 196830 & 196830 & 649539 & 649539 \\
         \# Comps & 6 & 6 & 7 & 7 & 8 & 8 & 9 & 9 & 10 & 10 & 11 & 11  \\
         \# Trans & 3240 & 3240 & 13122 & 13122 & 51030 & 51030 & 192456 & 192456 & 708588 & 708588 & 2558790 & 2558790 \\
         Bound & 30 & 30 & 30 & 30 & 30 & 30 & 30 & 30 & 30 & 30 & 30 & 30 \\ \midrule 	
        
		 {tool} & \multicolumn{10}{c}{{computation time in second}} \\ \midrule
		  \midrule
          \textsc{SpaceEx} & 0.059 & Timeout & 0.068 & Timeout & 0.079 & Timeout & 0.09 & Timeout & 0.103 & Timeout & 0.116 & Timeout \\
          \hline
          \textsc{dReach} & 23.5 & Timeout & 367.44 & Timeout & OOM & OOM & OOM & OOM & OOM & OOM & OOM & OOM \\
          \hline
          \textsc{SAT-Reach} & 13.26 & 7.59 & 105.21 & 19.86 & 1049.48 & 157.43 & Timeout & Timeout & Timeout & Timeout & Timeout & Timeout \\
          \hline
          \textsc{XSpeed} & 1.16 & Timeout & 1.54 & Timeout & 1.98 & Timeout & 2.48 & Timeout & Timeout & Timeout & Timeout & Timeout \\
  		 \hline
          \textsc{SAT-Reach-C} & $1.64$ & $7.09$ & $2.18$ & $9.96$ & $2.82$ & $13.29$ & $3.61$ & $17.7$ & $4.56$ & $22.86$ & $6.42$ & $29.78$ \\
          \hline
          \textsc{BACHopt} & 0.16 & 0.03 & 0.25 & 0.04 & 0.38 & 0.05 & 0.38 & 0.06 & 0.45 & 0.08 & 0.62 & 0.08 \\
          \midrule 
	\end{tabular}
	\label{tab:compTimes:nrs}
 \end{adjustbox}
\end{table}

\begin{table*}[htb]
	\setlength{\tabcolsep}{4pt}
	\renewcommand{\arraystretch}{1.2}
	\centering
	\caption{Computation times of the FDDI benchmark. The $-$ in the table indicates that the Cartesian product of the CLHA could not be constructed within a reasonable time.}
    \begin{adjustbox}{max width=0.8\textwidth}
	\begin{tabular}[c]{lcccccccc}
	    \toprule
        {instance} & FDDIS05 & FDDIU05 & FDDIS06 & FDDIU06 & FDDIS07 & FDDIU07 & FDDIS08 & FDDIU08  \\ \midrule 
        safety & safe & unsafe & safe & unsafe & safe & unsafe & safe & unsafe \\
        \# Vars & 15 & 15 & 18 & 18 & 21 & 21 & 24 & 24  \\
         \# Locations & 7776 & 7776 & 46656 & 46656 & 279936 & 279936 & 1679616 & 1679616 \\
         \# Comp & 5 & 5 & 6 & 6 & 7 & 7 & 8 & 8 \\
         \# Trans & 22680 & 22680 & 164592 & 164592 & 1158624 & 1158624 & - & - \\ 
         Bound & 30 & 30 & 30 & 30 & 30 & 30 & 30 & 30 \\ \midrule 	
		{tool} & \multicolumn{8}{c}{{computation time in [s]}} \\ \midrule
		  \midrule
         \textsc{SpaceEx} & 0.006 & 0.817 & 0.008 & 3.909 & 0.011 & 21.968 & 0.014 & 142.311 \\
         \hline
         \textsc{dReach} & OOM & 1.89 & OOM & 11.31 & OOM & 83.03 & Timeout* & Timeout*\\
         \hline
         \textsc{SAT-Reach} & 232.07 & 0.54 & 231.73 & 9.86 & Timeout & Timeout & Timeout & Timeout \\
         \midrule 
         \textsc{XSpeed} & 0.01 & 0.39 & 0.02 & 2.13 & Timeout & Timeout & Timeout & Timeout \\
         \hline
         \textsc{SAT-Reach-C} & 0.48 & 0.09& 0.55 & 0.14 & 0.64 & 0.19 & 0.74 & 0.25\\
         \hline
         \textsc{BACHopt} & 0.09 & 0.1 & 0.11 & 0.09 & 0.12 & 0.26 & 0.14 & 0.17\\
  		 \bottomrule
	\end{tabular}
    \end{adjustbox}
    \label{tab:compTimes:fddi}
\end{table*}
\vspace{-10pt}
\subsubsection{Nuclear Reactor System}
The NRS~\cite{NRS} schedules the sequential insertion of multiple rods into heavy water to absorb the neutrons. The system comprises a controller and rods as its member automata. The safety specification states that the system should never reach a state where all rods are in the recovery phase and the controller is in the initial position. NRSUn and NRSSn represent the system with n rods, which are considered unsafe and safe, respectively, with respect to the safety specification.

\noindent Table \ref{tab:compTimes:nrs} shows the computational time required by each tool for the safe and unsafe instances. \textsc{SAT-Reach-C} and \textsc{BACHopt} are the tools that classify all the instances within a reasonable time. \textsc{BACHopt} differentiates itself from the \textsc{SAT-Reach-C} by first retrieving compositional paths from the discrete layer and subsequently verifying their feasibility using linear constraint solving. Linear dynamics lends to the reduction of feasibility checking to linear constraint solving, justifying its better performance. Many of the safe and unsafe instances timeout for \textsc{XSpeed} and \textsc{SAT-Reach}, because the tools spend considerable time in the construction of the product automaton. For the same instances, \textsc{SAT-Reach-C} terminates within the timeout since it does not explicitly compute the product automaton.

\subsubsection{FDDI}
The FDDI \cite{fddi, fddimodel} model is a standard data transmission protocol on fiber optic lines. The system consists of several stations for transmitting and receiving data. Each station in the system waits for the signal of the previous station to transmit data. The safety specification is that all stations are not in location $s4$. The instances FDDISn and FDDIUn represent the systems consisting of $n$ stations, which are considered safe and unsafe instances, respectively, with respect to the safety specification. Table \ref{tab:compTimes:fddi} shows that \textsc{BACHopt} is the best performer. 

\noindent Table \ref{tab:split-table} highlights the impact of shared reachability analysis. We observe a reduction in computation time across some instances. This improvement stems from the reuse of reachable states during symbolic reachability analysis, effectively avoiding redundant $postC$ computations. Consequently, the optimization enhances the scalability of the algorithm, making it more suitable for large and complex compositional systems (NAV25U2). Gray cells denote the smaller number of $postC$ operations needed to classify instances. 

\begin{table}[ht]
\caption{Comparison of \textsc{SAT-Reach-C} and \textsc{SAT-Reach-C} without optimizations. Time is the total time needed to classify an instance.}
\centering
\setlength{\tabcolsep}{4pt}
\renewcommand{\arraystretch}{1.1}
\begin{subtable}{0.48\textwidth}
\centering
\begin{adjustbox}{max width=0.95\textwidth}
\begin{tabular}{|c|c|c|c|c|}
\hline
\multirow{2}{*}{Instance} & \multicolumn{2}{|c|}{\textsc{SAT-Reach-C}} & \multicolumn{2}{|c|}{\textsc{SAT-Reach-C} without opt.} \\
\cline{2-5}
& Time & \#$postC$ & Time & \#$postC$ \\
\hline
FOU04 & 0.18 & 4 & 0.18 & 4 \\
\hline
FOU016 & 1.16 & 4 & 1.17 & 4 \\
\hline
FOU064 & 22.72 & 4 & 23.15 & 4 \\
\hline
NAV3C2 (U) & 0.50 & 7 & 0.44 & 7 \\
\hline
NAV3C3 (U) & 1.33 & 10 & 1.35 & 10 \\
\hline
NAV3C4 (U) & 2.55 & 13 & 2.54 & 13 \\
\hline
NAV3C5 (U) & 3.37 & 16 & 4.46 & 16 \\
\hline
NAV3C2 (S) & 0.22 & 5 & 0.19 & 5 \\
\hline
NAV3C3 (S) & 0.65 & \cellcolor[gray]{0.8}12 & 0.72 & 16 \\
\hline
NAV3C4 (S) & 3.39 & \cellcolor[gray]{0.8}32 & 4.54 & 65 \\
\hline
NAV3C5 (S) & 60.7 & \cellcolor[gray]{0.8}102 & 67.33 & 326 \\
\hline
NAV25U1 & 2.57 & 5 & 2.60 & 5 \\
\hline
NAV25U2 & 661.7 & \cellcolor[gray]{0.8}1532 & 3053.80 & 10528 \\
\hline
NAV25S1 & 4.61 & 5 & 5.54 & 10 \\
\hline
NAV25S2 & 298.77 & \cellcolor[gray]{0.8}219 & 416.27 & 738 \\
\hline
NAV20U1 & 14.91 & 48 & 15.21 & 52 \\
\hline
NAV20U2 & 89.93 & \cellcolor[gray]{0.8}218 & 127.48 & 694 \\
\hline
NAV20U3 & 13.94 & \cellcolor[gray]{0.8}61 & 14.23 & 66 \\
\hline
NAV20S1 & 2.39 & 9 & 2.42 & 9 \\
\hline
NAV20S2 & 12.37 & \cellcolor[gray]{0.8}36 & 12.44 & 37 \\
\hline
\end{tabular}
\end{adjustbox}
\end{subtable}
\hfill
\begin{subtable}{0.48\textwidth}
\centering
\begin{adjustbox}{max width=0.9\textwidth}
\begin{tabular}{|c|c|c|c|c|}
\hline
\multirow{2}{*}{Instance} & \multicolumn{2}{|c|}{\textsc{SAT-Reach-C}} & \multicolumn{2}{|c|}{\textsc{SAT-Reach-C} without opt.} \\
\cline{2-5}
& Time & \#$postC$ & Time & \#$postC$ \\
\hline
NRSS05 & 1.64 & \cellcolor[gray]{0.8}1 & 3.10 & 5 \\
\hline
NRSU05 & 7.09 & 11 & 7.17 & 11 \\
\hline
NRSS06 & 2.18 & \cellcolor[gray]{0.8}1 & 4.35 & 6 \\
\hline
NRSU06 & 9.96 & 13 & 10.5 & 13 \\
\hline
NRSS07 & 2.82 & \cellcolor[gray]{0.8}1 & 5.92 & 7 \\
\hline
NRSU07 & 13.29 & 15 & 13.58 & 15 \\
\hline
NRSS08 & 3.61 & \cellcolor[gray]{0.8}1 & 7.83 & 8 \\
\hline
NRSU08 & 17.7 & 17 & 17.85 & 17 \\
\hline
NRSS09 & 4.56 & \cellcolor[gray]{0.8}1 & 10.21 & 9 \\
\hline
NRSU09 & 22.86 & 19 & 23.08 & 19 \\
\hline
NRSS010 & 6.42 & \cellcolor[gray]{0.8}1 & 13.64 & 10 \\
\hline
NRSU010 & 29.78 & 21 & 30.3 & 21 \\
\hline
FDDIS05 & 0.48 & 2 & 0.71 & 2 \\
\hline
FDDIU05 & 0.09 & 7 & 0.08 & 7 \\
\hline
FDDIS06 & 0.55 & 2 & 0.82 & 2 \\
\hline
FDDIU06 & 0.14 & 8 & 0.13 & 8 \\
\hline
FDDIS07 & 0.64 & 2 & 0.95 & 2 \\
\hline
FDDIU07 & 0.19 & 9 & 0.18 & 9 \\
\hline
FDDIS08 & 0.74 & 2 & 1.08 & 2 \\
\hline
FDDIU08 & 0.25 & 10 & 0.25 & 10 \\
\hline
\end{tabular}
\end{adjustbox}
\end{subtable}
\label{tab:split-table}
\end{table}

%% file: conclusion.tex
\section{Conclusion}
A CEGAR-based bounded reachability analysis algorithm has been proposed for compositional hybrid systems with piecewise affine dynamics. The proposed approach avoids explicit construction of the parallel product automaton by performing counterexample search in the discrete abstraction of the automaton and refining the state-space by leveraging symbolic reachability analysis. Furthermore, the algorithm employs mixed compositional semantics, step semantics during abstract counterexample generation, and shallow semantics during state-space refinement for scalability. Optimizations, such as caching the computed reachable states for reuse, have been proposed. The algorithm is implemented in the tool $\textsc{SAT-Reach}$ and results demonstrate the scalability benefits compared to the state-of-the-art.